\documentclass[11pt,a4paper,draft]{article}

\usepackage[ngerman,english]{babel}
\usepackage[utf8]{inputenc}

\usepackage[DIV=12]{typearea}

\usepackage[affil-it]{authblk}

\usepackage{xparse}
\usepackage{amsmath, amsfonts, amssymb, amsthm}
\usepackage{mathrsfs}
\usepackage{mathtools}
\usepackage{enumerate}
\usepackage{color}

\usepackage[noadjust]{cite}

\usepackage[capitalise]{cleveref}

\crefname{equation}{}{}

\newcommand{\R}{\mathbb{R}}

\newcommand{\N}{\mathbb{N}}

\newcommand{\ee}{\mathrm{e}}

\DeclareDocumentCommand\dd{ o g d() }{
	\IfNoValueTF{#2}{
		\IfNoValueTF{#3}
			{\mathrm{d}\IfNoValueTF{#1}{}{^{#1}}}
			{\mathinner{\mathrm{d}\IfNoValueTF{#1}{}{^{#1}}\argopen(#3\argclose)}}
		}
		{\mathinner{\mathrm{d}\IfNoValueTF{#1}{}{^{#1}}#2} \IfNoValueTF{#3}{}{(#3)}}
	}

\newcommand{\dx}{\dd{x}}
\newcommand{\dy}{\dd{y}}

\newcommand{\ds}{\dd{s}}
\newcommand{\dt}{\dd{t}}

\newcommand{\del}{\partial}
\newcommand{\eps}{\varepsilon}

\newcommand{\mom}{\mathfrak{m}}
\newcommand{\s}{\mathfrak{s}}

\DeclareMathOperator{\sgn}{sgn}

\newcommand{\vcc}{\vcentcolon}

\DeclarePairedDelimiter\abs{\lvert}{\rvert}
\DeclarePairedDelimiter\norm{\Vert}{\rVert}

\theoremstyle{plain}
\newtheorem{theorem}{Theorem}[section]
\newtheorem{lemma}[theorem]{Lemma}
\newtheorem{proposition}[theorem]{Proposition}

\theoremstyle{definition}
\newtheorem{definition}[theorem]{Definition}

\theoremstyle{remark}
\newtheorem{remark}[theorem]{Remark}

\numberwithin{equation}{section}

 \title{Smoluchowski's discrete coagulation equation with forcing}

 \author{Christian Kuehn\thanks{Technical University of Munich, Faculty of Mathematics, Research Unit 
``Multiscale and Stochastic Dynamics'', 85748 Garching b.~M\"unchen, Germany, \texttt{ckuehn@ma.tum.de}} }
 \author{Sebastian Throm\thanks{Technical University of Munich, Faculty of Mathematics, Research Unit 
``Multiscale and Stochastic Dynamics'', 85748 Garching b.~M\"unchen, Germany, \texttt{throm@ma.tum.de}} }
 \affil{}
 
 \date{}

\begin{document}

\maketitle

\begin{abstract}
In this article we study an extension of Smoluchowski's discrete coagulation equation, where particle in- and output takes place. This model is frequently used to describe aggregation processes in combination with sedimentation of clusters. More precisely, we show that the evolution equation is well-posed for a large class of coagulation kernels and output rates. Additionally, in the long-time limit we prove that solutions converge to a unique equilibrium with exponential rate under a suitable smallness condition on the coefficients.
\end{abstract}


\section{Introduction}\label{Sec:intro}

\subsection{Forced coagulation and coagulation-fragmentation}

In this article we consider the discrete coagulation equation
\begin{equation}\label{eq:Smol:1}
 \frac{\dd}{\dt}c_{k}=\frac{1}{2}\sum_{\ell=1}^{k-1}a_{k-\ell,\ell}c_{k-\ell}c_{\ell}-c_{k}\sum_{\ell=1}^{\infty}a_{k,\ell}c_{\ell}+s_{k}-r_{k}c_{k},\qquad c_k=c_k(t),
\end{equation}
which is used to describe the time evolution of a system of aggregating particles under the effect of external forcing. In particular, this equation is frequently used in cloud physics \cite{Fri00}, in oceanography \cite{BoG98}, and in chemistry \cite{HeZ85, AMY12}. More precisely, \eqref{eq:Smol:1} is an extension of Smoluchowski's original model \cite{Smo17} which reads
\begin{equation}\label{eq:coag}
 \frac{\dd}{\dt}c_{k}=\frac{1}{2}\sum_{\ell=1}^{k-1}a_{k-\ell,\ell}c_{k-\ell}c_{\ell}-c_{k}\sum_{\ell=1}^{\infty}a_{k,\ell}c_{\ell},
\end{equation}
and which corresponds to the choice $s_{k}=0$ and $r_{k}=0$ in~\eqref{eq:Smol:1}. 

The interpretation of~\eqref{eq:Smol:1} is the following. The quantity $c_{k}$ represents the density of particles of size/mass $k\in\N$ in the system, while we assume here that each cluster consists of a certain number of atoms. The time evolution of $c_{k}$ is then on the one hand determined by the production of clusters of size $k$ due to the coagulation of particles of sizes $k-\ell$ and $\ell$. This effect is taken into account by the first sum on the right-hand side of~\eqref{eq:Smol:1} and the symmetry of the coagulation process leads to the factor $1/2$. Conversely, clusters of size $k$ may aggregate with clusters of any size $\ell$ to form larger particles and this results in a decrease of $c_{k}$ which is considered by the second sum on the right-hand side of~\eqref{eq:Smol:1}. Additionally, in contrast to~\eqref{eq:coag}, we allow in~\eqref{eq:Smol:1} also that clusters of size $k$ are injected into the system with rate $s_{k}$. On the other hand, clusters are removed from the system with a rate given by $r_{k}c_{k}$. In the applications mentioned before, this removal term often corresponds to sedimentation of particles due to gravity.

Let us consider some important examples for the coefficients $a_{k,\ell}$ and $r_{k}$ appearing in~\eqref{eq:Smol:1}. In Smoluchowski's derivation of the coagulation equation in~\cite{Smo17}, he assumed that the particles in the system move freely according to Brownian motion and aggregate immediately once they touch. In this situation he obtained the following coagulation coefficient
\begin{equation}\label{eq:kernel:BM}
 a_{k,\ell}=(k^{1/3}+\ell^{1/3})(k^{-1/3}+\ell^{-1/3}).
\end{equation}
Another important example is given by
\begin{equation}\label{eq:kernel:shear}
 a_{k,\ell}=(k^{1/3}+\ell^{1/3})^{3}
\end{equation}
which models coagulation due to linear shear flow \cite{Ald99}. In the mathematical literature kernels of the form
\begin{equation}\label{eq:kernel:math}
 a_{k,\ell}=k^{\alpha}\ell^{\beta}+k^{\beta}\ell^{\alpha}
\end{equation}
are frequently used \cite{Hay87,HeZ85}. As already mentioned above, the removal coefficient $r_{k}$ typically models the effect of sedimentation. In this situation, we have the scaling $r_{k}\sim k^{\gamma}$ where the exponent $\gamma$ is related to the fractal dimension $D$ of the aggregates through $\gamma=1-1/D$ (see~\cite{BoG98}). More precisely, $k^{\gamma}$ corresponds to the terminal settling velocity which is determined by Stokes flow. Assuming, as usual, that all clusters are spheres which are indexed by their total mass/volume $k\in\N$, we obtain the scaling $r_{k}\sim k^{2/3}$ \cite{Fri00}. In~\cite{LiC71} the more sophisticated relation $r_{k}=C k^{2/3}[1+(0.084+0.0264 \ee^{-16.7 k^{1/3}})/(k^{1/3})]$ can be found. 

Concerning the injection rate $s_{k}$, a typical assumption is that only monomers are introduced, i.e.\@ $s_{k}=0$ for all $k>1$ (e.g.\@ \cite{Hay87,HeZ85,LuP76}). In this, work we will allow for more general sources. In fact, we will only require that the sequence $s_{k}$ decreases sufficiently fast for $k\to\infty$ (see the assumption~\eqref{eq:Ass:s} below).

Another intrinsic motivation to study~\eqref{eq:Smol:1} is to work towards a coupling with additional differential equations, which is a common theme in the context of reaction-diffusion systems. The idea is that $r_k$ and $s_k$, instead of being fixed, could be driven themselves. Yet, before this extension can be achieved, one should understand the case~\eqref{eq:Smol:1}.

A further important extension of Smoluchowski's coagulation model is the well-known coagulation-fragmentation equation which reads
\begin{equation}\label{eq:coag:frag}
 \frac{\dd}{\dt}c_{k}=\frac{1}{2}\sum_{\ell=1}^{k-1}a_{k-\ell,\ell}c_{k-\ell}c_{\ell}-c_{k}\sum_{\ell=1}^{\infty}a_{k,\ell}c_{\ell}+\sum_{\ell=1}^{\infty}b_{k,\ell}c_{k+\ell}-\frac{1}{2}\sum_{\ell=1}^{k-1}b_{k,k-\ell}c_{k}.
\end{equation}
In this model, contrary to the classical case~\eqref{eq:coag}, clusters are additionally allowed to split in smaller pieces and the interpretation is then analogous to the pure coagulation equation. Namely, the third sum on the right-hand side of~\eqref{eq:coag:frag} counts the particles of size $k$ which are created due to the breakup of a cluster of size $k+\ell$ while the fourth sum accounts for the loss of aggregates of size $k$ due to fragmentation.

A crucial question in all of the three models~\cref{eq:Smol:1,eq:coag,eq:coag:frag} concerns the long-time behaviour of the solutions. In comparison to~\cref{eq:Smol:1,eq:coag:frag}, Smoluchowski's original equation~\eqref{eq:coag} represents a special case since here clusters can only grow while no smaller particles are created or inserted in the system. Thus, the existence of a stationary state cannot be expected but the long-time behaviour is conjectured to be self-similar. This is known as the \emph{scaling hypothesis} \cite{LaM04} but, except for special \emph{solvable kernels} for which solutions can be computed explicitly, this conjecture is still unproven. 

In contrast to this, for~\cref{eq:coag:frag,eq:Smol:1} under reasonable assumptions on $b_{k,\ell}$ or $s_{k}$ and $r_{k}$ it is natural to expect that the system approaches an equilibrium state as $t\to\infty$. Although~\eqref{eq:coag:frag} has been studied intensively over the last decades, such a convergence could not been established for general coefficients $a_{k,\ell}$ and $b_{k,\ell}$. In fact, convergence to equilibrium for~\eqref{eq:coag:frag} is usually established under the \emph{detailed balance condition} where $a_{k,\ell}$ and $b_{k,\ell}$ are related due to
\begin{equation}\label{eq:detailed:balance}
 a_{k,\ell}Q_{k}Q_{\ell}=b_{k,\ell}Q_{k+\ell}
\end{equation}
for a non-negative $Q=(Q_{k})_{k\in\N}$ satisfying $Q\not\equiv 0$ and $\sum_{k=1}^{\infty}k Q_{k}<\infty$ (see e.g.\@ \cite{Car92,CaC94,Can07}). Without~\eqref{eq:detailed:balance} much less is known on the long-time behaviour of solutions to~\eqref{eq:coag:frag}. Exceptions are given on the one hand by~\cite{DuS96a}, where solvable kernels are considered for which explicit computations can be performed. On the other hand, in~\cite{FoM04} convergence to equilibrium for~\eqref{eq:coag:frag} has been shown without assuming the detailed balance condition. The proof relies on a functional inequality, which can be derived under a smallness condition on the first moment of the initial data and the assumption that fragmentation dominates coagulation in a suitable sense. The latter condition in known as \emph{strong fragmentation}. 

From a mathematical point of view, solutions to~\eqref{eq:Smol:1} and~\eqref{eq:coag:frag} exhibit a similar behaviour. We will in fact extend the method developed in~\cite{FoM04} to show convergence to a unique equilibrium in~\eqref{eq:Smol:1}. More precisely, under the conditions on $r_{k}$ which we impose (see~\eqref{eq:Ass:r}), the removal term $r_{k}c_{k}$ causes that solutions to~\eqref{eq:Smol:1} exhibit analogous properties as solutions to~\eqref{eq:coag:frag} in the strong fragmentation regime. However, there is also a fundamental difference between~\eqref{eq:Smol:1} and~\eqref{eq:coag:frag}. Precisely, at least in the situation of strong fragmentation, solutions to~\eqref{eq:coag:frag} conserve the total mass, i.e.\@ the first moment $\mom_{1}(t)\vcc=\sum_{k=1}^{\infty}kc_{k}(t)$ \cite{Cos95a,ELM03}. Conversely, such a behaviour cannot be expected for~\eqref{eq:Smol:1} where $\mom_{1}$ in general satisfies the equation (see~\eqref{eq:first:moment})
\begin{equation*}
 \frac{\dd}{\dt} \mom_{1}=\sum_{k=1}^{\infty}ks_{k}-\sum_{k=1}^{\infty}kr_{k}c_{k}(t).
\end{equation*}

In contrast to~\eqref{eq:coag:frag}, the model~\eqref{eq:Smol:1} has been studied much less intensively in the mathematical literature. Let us recapitulate some of the most important results. In~\cite{Spo85} local (in time) existence of solutions is shown for~\eqref{eq:Smol:1} under the condition that $\sum_{k=1}^{\infty}ks_{k}<\infty$ and $a_{k,\ell}\leq d_{k}d_{\ell}$ with $d_{n}/n\to 0$ as $n\to\infty$. Some explicit formulas for solvable coefficients $a_{k,\ell}$ and $r_{k}$ can be found in~\cite{Hen84,HeZ85}. Moreover, in~\cite{CRW06,CPS07} for the special choice $r_{k}\equiv 0$ (i.e.\@ no output), $a_{k,1}=a_{1,k}=1$ and $a_{k,\ell}=0$, $s_{k}=0$ if $k,\ell>1$ the long-time behaviour of solutions to~\eqref{eq:Smol:1} has been shown to be self-similar. Some formal considerations can be found in~\cite{Hay87}. The existence of stationary solutions to~\eqref{eq:Smol:1} as well as their uniqueness has been considered in~\cite{CrS82,Whi82}. 

In this work, we will establish the well-posedness of~\eqref{eq:Smol:1} for a large class of coefficients and moreover, we will establish convergence to a unique steady state under an additional smallness condition.

\subsection{Assumptions on the coefficients and main results}

In this section, we collect the assumptions on the coefficients of~\eqref{eq:Smol:1} and we state the main results, which we are going to show. 

We assume that the coagulation kernel $a_{k,\ell}$ is symmetric, non-negative and satisfies a suitable growth condition, i.e.\@ we have
\begin{equation}\label{eq:Ass:a}
 a_{k,\ell}=a_{\ell,k}\quad  \text{and}\quad 0\leq a_{k,\ell}\leq A_{*} (k^{\alpha}\ell^{\beta}+k^{\beta}\ell^{\alpha})\quad \forall k,\ell\in\N \quad \text{with } \alpha,\beta\in[0,1],~\alpha\leq \beta,
\end{equation}
where $A_*>0$ is a constant. This covers in particular the two examples~\cref{eq:kernel:BM,eq:kernel:shear}. Moreover, we assume that the removal coefficients grow sufficiently fast for large cluster sizes, i.e.\@
\begin{equation}\label{eq:Ass:r}
 r_{k}\geq R_{*} k^{\gamma}\quad \text{for all }k\in\N \quad \text{with }R_{*}>0 \text{ and }\gamma>\max\{0,\alpha+\beta-1\}.
\end{equation}
Note that the condition on the exponent $\gamma$ is crucial in our analysis in order to obtain suitable moment estimates. We emphasise that this includes in particular the important example of sedimentation with particles coagulating due to Brownian motion, i.e.\@ $r_{k}=k^{2/3}$ and $a_{k,\ell}$ as in~\eqref{eq:kernel:BM}.

Finally, we require that the source term $s_{k}$ is non-negative and decays faster than any power law, i.e.\@ we assume that
\begin{equation}\label{eq:Ass:s}
 s_{k}\geq 0\quad \text{and for each }\mu\geq 0 \text{ there exists }\s_{\mu}>0\text{ such that } \sum_{k=1}^{\infty}k^{\mu}s_{k}\leq \s_{\mu}.
\end{equation}
The latter condition is not really a restriction in typical applications since usually there are no huge clusters injected into the system, i.e.\@ $s_{k}$ is constantly zero for large $k$.

For later use, we also introduce the scaled coefficients
\begin{equation}\label{eq:reduce:coefficients}
 \widehat{A}_{*}\vcc=\frac{A_{*}}{R_{*}},\qquad \text{and}\qquad \widehat{\s}_{\mu}=\frac{\s_{\mu}}{R_{*}}\quad \text{for all } \mu\geq 0.
\end{equation}
Moreover, for $\mu\geq 0$ we use the notation
\begin{equation*}
 \ell_{\mu}^{1}\vcc=\Bigl\{(c_k)_{k\in\N}\;\Big|\; c_{k}\in[0,\infty) \text{ for all }k\in\N\text{ and }\sum_{k=1}^{\infty}k^{\mu}c_{k}<\infty\Bigr\}
\end{equation*}
for the weighted $\ell^{1}$-spaces. 

Throughout this work we use the following notion of solutions, which is adapted from~\cite{FoM04}. 

\begin{definition}\label{Def:solution}
 A sequence $c(t)=(c_{k}(t))_{k\in\N}$ with $c_{k}\colon[0,T)\to[0,\infty)$ continuous for all $k\in\N$ is a solution to~\eqref{eq:Smol:1} on $[0,T)$ with initial condition $c^{\text{in}}=(c^{\text{in}}_{k})_{k\in\N}$ provided
 \begin{enumerate}[(i)]
  \item for each $k\in\N$ the equation~\eqref{eq:Smol:1} is satisfied for all $t\in(0,T)$,
  \item for each $\mu\geq 1$ we have $c\in L^{\infty}([0,T),\ell_{1}^{1})\cap C^{1}((0,T),\ell_{\mu}^{1})$,
  \item $c_{k}(0)=c_{k}^{\text{in}}$ for all $k\in\N$.
 \end{enumerate}
 We call the sequence $c$ a global solution if $T=\infty$. Moreover, $c$ is an equilibrium of~\eqref{eq:Smol:1} if $c$ is a stationary (global) solution of~\eqref{eq:Smol:1}, i.e.\@ $c$ does not depend on $t$.
\end{definition}

\begin{remark}
 In particular, we consider throughout this work only initial data $c^{\text{in}}$ satisfying $c^{\text{in}}\in\ell_{1}^{1}$.
\end{remark}

The first main result that we will show in this work is the existence of solutions to~\eqref{eq:Smol:1}.

\begin{theorem}\label{Thm:existence:evolution}
 Let \cref{eq:Ass:a,eq:Ass:r,eq:Ass:s} be satisfied and let $c^{\text{in}}=(c_{k}^{\text{in}})_{k\in\N}\in\ell_{1}^{1}$. Then there exists at least one solution $c=(c_{k})_{k\in\N}$ of~\eqref{eq:Smol:1} with initial condition $c^{\text{in}}$.
\end{theorem}

Under further restrictions on either the exponents $\beta$ and $\gamma$ or on the coefficients $A_{*}$, $\s_{1}$ and $R_{*}$  we also obtain uniqueness of solutions

\begin{proposition}\label{Prop:uniqueness}
 Let \cref{eq:Ass:a,eq:Ass:r,eq:Ass:s} be satisfied. If either $\beta<\gamma$ or alternatively the conditions of Theorem~\ref{Thm:existence:stat:sol} below are satisfied then there exists at most one solution to~\eqref{eq:Smol:1}.
\end{proposition}

A further main result we are going to prove concerns the existence of a unique equilibrium under a smallness condition on the coefficients. We note that the latter is mainly important for the uniqueness part. Concerning existence, one could in fact obtain a much stronger result (see also Remark~\ref{Rem:equilibria}).

\begin{theorem}\label{Thm:existence:stat:sol}
 Let \cref{eq:Ass:a,eq:Ass:r,eq:Ass:s} be satisfied and for $\mu\geq 1$ such that $\mu+\beta> \max\{2-\alpha-\beta,1\}$ assume that either
 \begin{align*}
  &16C_{\mu}A_{*}\Bigl(\frac{(2^{\mu+\beta}(\mu+\beta))^{p}q^{1-p}}{p}\widehat{A}_{*}^{p}\widehat{\s}_{1}^{1+p}+\widehat{\s}_{\mu+\beta}\Bigr)-R_{*}<0
\shortintertext{or}
 &8C_{\mu}A_{*}\Bigl(\frac{2^{2+\rho_{\mu+\beta}}(2^{\mu+\beta}(\mu+\beta))^{p}q^{1-p}}{p} \widehat{A}_{*}^{p}\widehat{\s}_{1}^{1+p+\rho_{\mu+\beta}}+2^{2+\rho_{\mu+\beta}}\widehat{\s}_{\mu+\beta}\widehat{\s}_{1}^{\rho_{\mu+\beta}}\Bigr)^{\frac{1}{1+\rho_{\mu+\beta}}}-R_{*}<0,
\end{align*}
where $C_{\mu}\vcc=2^{\max\{\mu-2,0\}}\max\{\mu,\mu(\mu-1)\}$, $\rho_{\mu+\beta}=\gamma/(\mu+\beta-1)$, $p=(\mu +\gamma-1)/(1+\gamma-\alpha-\beta)$ and $q=p/(p-1)$. Then there exists a unique stationary solution $Q=(Q_{k})_{k\in\N}$ of~\eqref{eq:Smol:1} and $Q\in \ell_{\nu}^{1}$ for all $\nu\geq 0$.
\end{theorem}

The assumptions in Theorem~\eqref{Thm:existence:stat:sol} essentially amount to a sufficient balance of injection versus removal of clusters. Moreover, we also obtain that each solution to~\eqref{eq:Smol:1} converges to the unique equilibrium.

\begin{theorem}\label{Thm:convergence}
 Let the assumptions of Theorem~\ref{Thm:existence:stat:sol} be satisfied. Then, for each solution $c=(c_{k})_{k\in\N}$ of~\eqref{eq:Smol:1} there exists $T_{c}>0$ such that
 \begin{equation*}
  \norm{c(t)-Q}_{\ell_{\mu}^{1}}\leq K\ee^{-\kappa (t-T_{c})} \quad \text{for all }t\geq T_{c}
 \end{equation*}
 where the constants $K,\kappa>0$ are independent of $c$.
\end{theorem}

\begin{remark}\label{Rem:equilibria}
 Note that the existence of stationary solutions to~\eqref{eq:Smol:1} for the case $\beta=\alpha+\gamma$ and $\alpha<1$ has already been proven in~\cite{CrS82} without the smallness condition which we assume in Theorem~\ref{Thm:existence:stat:sol}. However, in~\cite{CrS82} no uniqueness is established and, to our knowledge, convergence of solutions to~\eqref{eq:Smol:1} to these equilibria has not yet been shown.
\end{remark}

\subsection{Outline of the article}

The remainder of this work is organised as follows. In Section~\ref{Sec:moment:estimates} we establish several estimates on moments of solutions to~\eqref{eq:Smol:1}. Based on these moment bounds, we will then give the proof of Theorem~\ref{Thm:existence:evolution} in Section~\ref{Sec:existence}. Following the approach of~\cite{FoM04}, in Section~\ref{Sec:func:inequality} we show an important functional inequality. This inequality be key for the proofs of \cref{Thm:existence:stat:sol,Thm:convergence}, which are then contained in Section~\ref{Sec:convergence}. Finally, in Section~\ref{Sec:example} we provide an example, which illustrates that we cannot expect that our proof to show convergence to equilibrium can be extended to coefficients without a smallness condition as in Theorem~\ref{Thm:existence:stat:sol}.

\section{Moment estimates}\label{Sec:moment:estimates}

In this section we provide several a priori estimates for the moments of solutions to~\eqref{eq:Smol:1}. To simplify the presentation, we use for $\mu\in[0,\infty)$ and a fixed solution $c=(c_{k})_{k\in\N}$ of~\eqref{eq:Smol:1} with initial condition $c^{\text{in}}=(c_{k}^{\text{in}})_{k\in\N}$ the notation
\begin{equation*}
 \mom_{\mu}(t)=\sum_{k=1}^{\infty}k^{\mu}c_{k}(t),\quad \mom_{\mu}^{N}(t)=\sum_{k=1}^{N}k^{\mu}c_{k}(t)\quad\text{and}\quad \mom_{1}^{\text{in}}=\sum_{k=1}^{\infty}kc_{k}^{\text{in}}.
\end{equation*}
while we also note that $\mom_{1}^{\text{in}}=\mom_{1}(0)$. Note that $\mom_{\mu}^{N}$ exists for each $\mu\in[0,\infty)$ since in this case the sum is finite.

\begin{remark}\label{Rem:moments:monotonicity}
Since we are dealing with the discrete coagulation equation, we note that the non-negativity of the solutions immediately yields that
 \begin{equation*}
  \mom_{\mu_{1}}\leq \mom_{\mu_{2}},\quad\text{and}\quad \mom_{\mu_{1}}^{N}\leq \mom_{\mu_{2}}^{N}
 \end{equation*}
 if $0\leq \mu_{1}\leq \mu_{2}<\infty$.
\end{remark}

To get bounds on the moments we note the following well-known relation for the coagulation operator. For a sequence $\varphi=(\varphi_{k})_{k\in\N}$ with at most polynomial growth and each solution $c=(c_{k})_{k\in\N}$ of~\eqref{eq:Smol:1} we have due to the symmetry of $a_{k,\ell}$ that
\begin{equation}\label{eq:weak:coag:op}
 \sum_{k=1}^{\infty}\varphi_{k}\biggl(\frac{1}{2}\sum_{\ell=1}^{k-1}a_{k-\ell,\ell}c_{k-\ell,\ell}c_{\ell}-c_{k}\sum_{\ell=1}^{\infty}a_{k,\ell}c_{\ell}\biggr)=\frac{1}{2}\sum_{k=1}^{\infty}\sum_{\ell=1}^{\infty}a_{k,\ell}c_{k}c_{\ell}(\varphi_{k+\ell}-\varphi_{k}-\varphi_{\ell}).
\end{equation}
Thus, multiplying~\eqref{eq:Smol:1} with $\varphi_{k}$ and summing over $k\in\N$ we find
\begin{equation}\label{eq:Smol:weak}
 \frac{\dd}{\dd{t}}\sum_{k=1}^{\infty}c_{k}\varphi_{k}=\frac{1}{2}\sum_{k=1}^{\infty}\sum_{\ell=1}^{\infty}a_{k,\ell}c_{k}c_{\ell}[\varphi_{k+\ell}-\varphi_{k}-\varphi_{\ell}]+\sum_{k=1}^{\infty}s_{k}\varphi_{k}-\sum_{k=1}^{\infty}r_{k}c_{k}\varphi_{k}.
\end{equation}
Let $\chi_\mathcal{S}$ denote the indicator function for a set $\mathcal{S}$. If we choose $\varphi_{k}=k\chi_{\{k\leq N\}}$ in~\eqref{eq:Smol:weak}, we obtain together with $s_{k}\geq 0$ that
\begin{equation}\label{eq:first:moment}
 \frac{\dd}{\dt}\mom_{1}^{N}(t)=\sum_{k=1}^{N}ks_{k}-\sum_{k=1}^{N}kr_{k}c_{k}\leq \s_{1}-\sum_{k=1}^{N}kr_{k}c_{k}.
\end{equation}
Due to assumption~\eqref{eq:Ass:r} we have $r_{k}\geq R_{*}$ for all $k\in\N$ which yields
\begin{equation}\label{eq:first:moment:inequality}
 \frac{\dd}{\dt}\mom_{1}^{N}(t)\leq \s_{1}-R_{*}\mom_{1}^{N}.
\end{equation}
From this inequality we already conclude the following statement.

\begin{lemma}\label{Lem:total:mass:1}
 Let $c=(c_{k})_{k\in\N}$ be a solution to~\eqref{eq:Smol:1} with initial condition $c^{\text{in}}$. Then the corresponding first moment $\mom_{1}$ is uniformly bounded. More precisely, we have
 \begin{equation*}
  \mom_{1}(t)=\sum_{k=1}^{\infty}kc_{k}(t)\leq \max\{\mom_{1}^{\text{in}},\widehat{\s}_{1}\} \quad \text{for all }t\geq 0.
 \end{equation*}
 Moreover, there exists a time $T>0$ which depends only on $\mom_{1}^{\text{in}}$, $\widehat{\s}_{1}$ and $R_{*}$ such that $\mom_{1}(t)\leq 2\widehat{\s}_{1}$ for all $t\geq T$. 
\end{lemma}

\begin{proof}
 For fixed $N\in\N$ we apply Grönwall's inequality to~\eqref{eq:first:moment:inequality} which yields
 \begin{equation}\label{eq:total:mass:1}
  \mom_{1}^{N}(t)\leq \Bigl(\mom_{1}^{N}(0)-\frac{\s_{1}}{R_{*}}\Bigr)\ee^{-R_{*}t}+\frac{\s_{1}}{R_{*}}=(\mom_{1}^{N}(0)-\widehat{\s}_{1})\ee^{-R_{*}t}+\widehat{\s}_{1}\quad \text{for all }t\geq 0.
 \end{equation}
From this we easily deduce that $\mom_{1}^{N}(t)\leq \max\{\mom_{1}^{N}(0),\widehat{\s}_{1}\}$ and using that $\mom_{1}^{N}(0)\leq \mom_{1}^{\text{in}}$ the first claim follows upon taking the limit $N\to\infty$. To prove the second claim, we note that for $\mom_{1}^{\text{in}}=0$ there is nothing to show. On the other hand, if $\mom_{1}^{\text{in}}>0$ we conclude immediately from~\eqref{eq:total:mass:1} that it suffices to take $T>\max\{0, \log(\widehat{\s}_{1}/\mom_{1}^{\text{in}})/R_{*}\}$.
\end{proof}

Based on the bound on $\mom_{1}^{N}$ we can now derive also uniform estimates on higher order moments. In fact, we take $\varphi_{k}=k^{\mu}\chi_{\{k\leq N\}}$ with $\mu>\max\{2-\alpha-\beta,1\}$ in~\eqref{eq:Smol:weak} and note that
\begin{equation*}
 (k+\ell)^{\mu}\chi_{\{k+\ell\leq N\}}-k^{\mu}\chi_{\{k\leq N\}}-\ell^{\mu}\chi_{\{\ell\leq N\}}\leq \bigl((k+\ell)^{\mu}-k^{\mu}-\ell^{\mu}\bigr)\chi_{\{k\leq N\}\cup \{\ell\leq N\}}.
\end{equation*}
 Thus, we conclude from~\eqref{eq:Smol:weak} together with~\cref{eq:Ass:a,eq:Ass:r,eq:Ass:s} that
\begin{equation}\label{eq:higher:moment:1}
 \begin{aligned}
   \frac{\dd}{\dt}\mom_{\mu}^{N}(t)&\leq\frac{1}{2}\sum_{k=1}^{N}\sum_{\ell=1}^{N}a_{k,\ell}c_{k}(t)c_{\ell}(t)\bigl[(k+\ell)^{\mu}-k^{\mu}-\ell^{\mu}\bigr]+\sum_{k=1}^{N}k^{\mu}s_{k}-\sum_{k=1}^{N}k^{\mu}r_{k}c_{k}(t)\\
   &\leq \frac{A_{*}}{2}\sum_{k,\ell=1}^{N}c_{k}(t)c_{\ell}(t)\bigl(k^{\alpha}\ell^{\beta}+k^{\beta}\ell^{\alpha}\bigr)\bigl[(k+\ell)^{\mu}-k^{\mu}-\ell^{\mu}\bigr]+\s_{\mu}-R_{*}\mom_{\mu+\gamma}^{N}.
 \end{aligned}
\end{equation}
In order to continue, we estimate the expression $\bigl(k^{\alpha}\ell^{\beta}+k^{\beta}\ell^{\alpha}\bigr)\bigl[(k+\ell)^{\mu}-k^{\mu}-\ell^{\mu}\bigr]$ following the same approach as in~\cite{EMR05}. If $\ell\leq k$ we define $z=\ell/k\leq 1$ and find
\begin{multline*}
 \bigl(k^{\alpha}\ell^{\beta}+k^{\beta}\ell^{\alpha}\bigr)\bigl[(k+\ell)^{\mu}-k^{\mu}-\ell^{\mu}\bigr]=k^{\alpha+\beta+\mu}\bigl(z^{\beta}+z^{\alpha}\bigr)\bigl[(1+z)^{\mu}-1-z^{\mu}\bigr]\\*
 \leq 2k^{\alpha+\beta+\mu}z^{\alpha}\bigl[(1+z)^{\mu}-1-z^{\mu}\bigr]\leq 2^{\mu}\mu k^{\alpha+\beta+\mu}z^{\alpha+\nu}\leq 2^{\mu}\mu k^{\beta+\mu-\nu}\ell^{\alpha+\nu}
\end{multline*}
for each $\nu\in[0,1]$ while we also used that $\alpha\leq \beta$. If $k\leq\ell$ we can argue in the same way due to symmetry if we interchange $k$ and $\ell$ which finally yields
\begin{equation*}
 \bigl(k^{\alpha}\ell^{\beta}+k^{\beta}\ell^{\alpha}\bigr)\bigl[(k+\ell)^{\mu}-k^{\mu}-\ell^{\mu}\bigr]\leq 2^{\mu}\mu\bigl[k^{\beta+\mu-\nu}\ell^{\alpha+\nu}+k^{\alpha+\nu}\ell^{\beta+\mu-\nu}\bigr].
\end{equation*}
If we use this estimate in~\eqref{eq:higher:moment:1} it follows due to symmetry that
\begin{equation*}
 \frac{\dd}{\dt}\mom_{\mu}^{N}(t)\leq (2^{\mu}\mu A_{*})\mom_{\beta+\mu-\nu}^{N}\mom_{\alpha+\nu}^{N}+\s_{\mu}-R_{*}\mom_{\mu+\gamma}^{N}.
\end{equation*}
We choose now $\nu=1-\alpha$ which yields
\begin{equation}\label{eq:higher:moment:2}
 \frac{\dd}{\dt}\mom_{\mu}^{N}(t)\leq (2^{\mu}\mu A_{*})\mom_{\alpha+\beta+\mu-1}^{N}\mom_{1}^{N}+\s_{\mu}-R_{*}\mom_{\mu+\gamma}^{N}.
\end{equation}
From Hölder's inequality with 
\begin{equation}\label{eq:Hoelder:exp:1}
 p=\frac{\mu+\gamma-1}{1+\gamma-\alpha-\beta}\quad \text{and}\quad  q=\frac{p}{p-1} 
\end{equation}
we deduce
\begin{equation}\label{eq:moment:Hoelder}
 \mom_{\alpha+\beta+\mu-1}^{N}=\sum_{k=1}^{N}k^{\alpha+\beta+\mu-1}c_{k}=\sum_{k=1}^{N}k^{1/p}k^{\alpha+\beta+\mu-1-1/p}c_{k}^{1/p}c_{k}^{1/q}\leq(\mom_{1}^{N})^{\frac{1}{p}}\bigl(\mom_{\mu+\gamma}^{N}\bigr)^{\frac{1}{q}}.
\end{equation}
If we plug this estimate in~\eqref{eq:higher:moment:2} we further get
\begin{equation*}
 \frac{\dd}{\dt}\mom_{\mu}^{N}(t)+R_{*}\mom_{\mu+\gamma}^{N}\leq (2^{\mu}\mu A_{*})(\mom_{1}^{N})^{1+\frac{1}{p}}(\mom_{\mu+\gamma}^{N})^{\frac{1}{q}}+\s_{\mu}.
\end{equation*}
Young's inequality (with $\eps$) yields that 
\begin{equation*}
 (\mom_{1}^{N})^{1+\frac{1}{p}}(\mom_{\mu+\gamma}^{N})^{\frac{1}{q}}\leq \eps\mom_{\mu+\gamma}^{N}+\frac{(q\eps)^{1-p}}{p}(2^{\mu}\mu A_{*})^{p}(\mom_{1}^{N})^{1+p}.
\end{equation*}
 Thus, taking $\eps=R_{*}/2$ we find
\begin{equation}\label{eq:higher:moment:3}
 \frac{\dd}{\dt}\mom_{\mu}^{N}(t)+\frac{1}{2} R_{*}\mom_{\mu+\gamma}^{N}\leq \frac{(2^{\mu-1}\mu)^{p}q^{1-p}}{2p}R_{*}^{1-p}A_{*}^{p}(\mom_{1}^{N})^{1+p}+\s_{\mu}.
\end{equation}
In the same way as in~\eqref{eq:moment:Hoelder} Hölder's inequality with
\begin{equation}\label{eq:Hoelder:exp:2}
 \widetilde{p}=\frac{\mu+\gamma-1}{\gamma}\quad \text{and}\quad \widetilde{q}=\frac{\mu+\gamma-1}{\mu-1}
\end{equation}
yields
\begin{equation*}
 \mom_{\mu}^{N}=\sum_{k=1}^{N}k^{\mu}c_{k}^{N}=\sum_{k=1}^{N}k^{\mu-\frac{1}{\widetilde{p}}}k^{\frac{1}{\widetilde{p}}}c_{k}^{\frac{1}{\widetilde{p}}}c_{k}^{\frac{1}{\widetilde{q}}}\leq (\mom_{1}^{N})^{\frac{1}{\widetilde{p}}}(\mom_{\mu+\gamma}^{N})^{\frac{1}{\widetilde{q}}}.
\end{equation*}
Together with~\eqref{eq:Hoelder:exp:2} this can be rearranged as
\begin{equation}\label{eq:moment:proof:1}
 (\mom_{\mu}^{N})^{\frac{\mu+\gamma-1}{\mu-1}}\leq (\mom_{1}^{N})^{\frac{\gamma}{\mu-1}} \mom_{\mu+\gamma}^{N}.
\end{equation}
These estimates enable us to derive a couple of moment estimates that will be essential for the rest of this work. One ingredient for this is the following nonlinear version of Grönwall's inequality (see also~\cite{FoM04,ELM03}).

\begin{lemma}\label{Lem:nonlinear:gronwall}
 Let $\rho>1$ and let $f\in C([t_{0},\infty),\R_{\geq 0})\cap C^{1}((t_{0},\infty))$ satisfy
 \begin{equation}\label{eq:nonlinear:gronwall:1}
  \frac{\dd}{\dt}f(t)+\Lambda (f(t))^{1+\rho}\leq \Xi \quad \text{for all }t>t_{0}
 \end{equation}
 with constants $\Lambda, \Xi>0$. Then
 \begin{equation*}
  f(t)\leq \max\Bigl\{\Bigl(\frac{2\Xi}{\Lambda}\Bigr)^{\frac{1}{1+\rho}},\Bigl(\frac{2}{\rho\Lambda}\Bigr)^{\frac{1}{\rho}}(t-t_{0})^{-\frac{1}{\rho}}\Bigr\}
 \end{equation*}
 for all $t>t_{0}$.
\end{lemma}

\begin{proof}
 This proof follows mainly arguments contained in~\cite{FoM04}. We define 
 \begin{equation*}
  T\vcc=\inf\Bigl\{t>t_{0}\;\Big|\; f(t)\leq \Bigl(\frac{2\Xi}{\Lambda}\Bigr)^{\frac{1}{1+\rho}}\Bigr\}.
 \end{equation*}
 The inequality~\eqref{eq:nonlinear:gronwall:1} directly implies that
 \begin{equation}\label{eq:nonlinear:Gronwall:2}
  f(t)\leq \Bigl(\frac{2\Xi}{\Lambda}\Bigr)^{\frac{1}{1+\rho}}\quad \text{for all }t\geq T.
 \end{equation}
 On the other hand, the definition of $T$ implies that
 \begin{equation*}
  \Xi\leq \frac{\Lambda}{2}\bigl(f(t)\bigr)^{1+\rho}\quad \text{if } t_{0}<t<T.
 \end{equation*}
 Thus, together with~\eqref{eq:nonlinear:gronwall:1} we deduce
 \begin{equation*}
  \frac{\dd}{\dt}f(t)+\frac{\Lambda}{2}\bigl(f(t)\bigr)^{1+\rho}\leq 0\quad \text{for }t\in(t_{0},T).
 \end{equation*}
Integrating this inequality explicitly, we obtain together with the non-negativity of $f(t_{0})$ that
\begin{equation*}
 f(t)\leq \biggl(\frac{1}{(f(t_{0}))^{-\rho}+\frac{\Lambda\rho}{2}(t-t_{0})}\biggr)^{\frac{1}{\rho}}\leq \Bigl(\frac{2}{\rho\Lambda}\Bigr)^{\frac{1}{\rho}}(t-t_{0})^{-\frac{1}{\rho}}\quad \text{for all }t\in(t_{0},T).
\end{equation*}
Together with~\eqref{eq:nonlinear:Gronwall:2} the claim immediately follows.
\end{proof}

Based on Lemma~\ref{Lem:nonlinear:gronwall} and~\eqref{eq:moment:proof:1} we can now prove the following estimate on higher order moments of solutions to~\eqref{eq:Smol:1}.

\begin{lemma}\label{Lem:moment:general}
 Let $c=(c_{k})_{k\in\N}$ be a solution to~\eqref{eq:Smol:1} with corresponding first moment $\mom_{1}$ and initial condition $c^{\text{in}}=(c_{k}^{\text{in}})_{k\in\N}$. For $\mu >\max\{2-\alpha-\beta,1\}$ let $\rho_{\mu }=\gamma/(\mu -1)$ and $p$, $q$ as in~\eqref{eq:Hoelder:exp:1}. Then 
 \begin{multline*}
   \mom_{\mu }(t)\leq \max\biggl\{\Bigl(\frac{2(2^{\mu -1}\mu )^{p}q^{1-p}}{p} \widehat{A}_{*}^{p}\max\{\mom_{1}^{\text{in}},\widehat{\s}_{1}\}^{1+p+\rho_{\mu }}+4\widehat{\s}_{\mu }\max\{\mom_{1}^{\text{in}},\widehat{\s}_{1}\}^{\rho_{\mu }}\Bigr)^{\frac{1}{1+\rho_{\mu }}},\\*
   \Bigl(\frac{4}{R_{*}\rho_{\mu }}\Bigr)^{\frac{1}{\rho_{\mu }}}\max\{\mom_{1}^{\text{in}},\widehat{s}_{1}\}t^{-\frac{1}{\rho_{\mu }}} \biggr\}
 \end{multline*}
for all $t\geq 0$. In particular, there exists a constant $C_{\mu }$ which only depends on the parameters $\alpha,\beta,\gamma,\mom_{1}^{\text{in}}, \widehat{s}_{1}, \widehat{\s}_{\mu }$, $\widehat{A}_{*}$ and $R_{*}$ such that $\mom_{\mu }(t)\leq C_{\mu }(1+t^{-1/\rho_{\mu }})$. 
\end{lemma}

\begin{proof}
Due to Lemma~\ref{Lem:total:mass:1} we have $\mom_{1}^{N}(t)\leq \max\{\mom_{1}^{\text{in}},\widehat{\s}_{1}\}$ uniformly in $N\in\N$ for all $t\geq 0$. Thus, we can estimate the right-hand side of~\eqref{eq:moment:proof:1} to obtain 
 \begin{equation*}
 (\mom_{\mu }^{N})^{\frac{\mu +\gamma-1}{\mu -1}}\leq \max\{\mom_{1}^{\text{in}},\widehat{\s}_{1}\}^{\frac{\gamma}{\mu -1}} \mom_{\mu +\gamma}^{N}.
\end{equation*}
Since $\widehat{\s}_{1}>0$ by assumption, we have $\max\{\mom_{1}^{\text{in}},\widehat{\s}_{1}\}>0$ and we find together with~\eqref{eq:higher:moment:3} that
\begin{equation*}
 \frac{\dd}{\dt}\mom_{\mu }^{N}(t)+\frac{1}{2} R_{*}\max\{\mom_{1}^{\text{in}},\widehat{\s}_{1}\}^{-\frac{\gamma}{\mu -1}}(\mom_{\mu }^{N})^{\frac{\mu +\gamma-1}{\mu -1}}\leq \frac{(2^{\mu -1}\mu )^{p}q^{1-p}}{2p}R_{*}^{1-p}A_{*}^{p}(\mom_{1}^{N})^{1+p}+\s_{\mu }.
\end{equation*}
 Using again $\mom_{1}^{N}(t)\leq \max\{\mom_{1}^{\text{in}},\widehat{\s}_{1}\}$ to estimate the right-hand side we finally obtain
 \begin{multline}\label{eq:moment:diffineq:2}
  \frac{\dd}{\dt}\mom_{\mu }^{N}(t)+\frac{1}{2} R_{*}\max\{\mom_{1}^{\text{in}},\widehat{\s}_{1}\}^{-\frac{\gamma}{\mu -1}}(\mom_{\mu }^{N})^{\frac{\mu +\gamma-1}{\mu -1}}\\*
  \leq \frac{(2^{\mu -1}\mu )^{p}q^{1-p}}{2p}R_{*}^{1-p}A_{*}^{p}\max\{\mom_{1}^{\text{in}},\widehat{\s}_{1}\}^{1+p}+\s_{\mu }.
 \end{multline}
 Thus, applying Lemma~\ref{Lem:nonlinear:gronwall} with $t_{0}=0$ and 
 \begin{equation}\label{eq:constants:1}
   \Xi\vcc=\frac{(2^{\mu -1}\mu )^{p}q^{1-p}}{2p}R_{*}^{1-p}A_{*}^{p}\max\{\mom_{1}^{\text{in}},\widehat{\s}_{1}\}^{1+p}+\s_{\mu }\quad \text{and}\quad
   \Lambda\vcc=\frac{1}{2} R_{*}\max\{\mom_{1}^{\text{in}},\widehat{\s}_{1}\}^{-\frac{\gamma}{\mu -1}}
 \end{equation}
 the claim follows if we recall~\eqref{eq:reduce:coefficients}.
\end{proof}

The next lemma states that after a sufficiently large time, the higher moments can be estimated independently of the initial data.

\begin{lemma}\label{Lem:moment:general:large:time}
 Let $c=(c_{k})_{k\in\N}$ be a solution to~\eqref{eq:Smol:1} with corresponding first moment $\mom_{1}$. For $\mu >\max\{2-\alpha-\beta,1\}$ let $\rho_{\mu }=\gamma/(\mu -1)$ and $p$, $q$ as in~\eqref{eq:Hoelder:exp:1}. Then there exists $T>0$ such that
 \begin{align}
   \mom_{\mu }(t)&\leq 2\Bigl(\frac{2^{2+\rho_{\mu }}(2^{\mu }\mu )^{p}q^{1-p}}{p} \widehat{A}_{*}^{p}\widehat{\s}_{1}^{1+p+\rho_{\mu }}+2^{2+\rho_{\mu }}\widehat{\s}_{\mu }\widehat{\s}_{1}^{\rho_{\mu }}\Bigr)^{\frac{1}{1+\rho_{\mu }}} \label{eq:large:time:moment:1} \shortintertext{and}
   \mom_{\mu }(t)&\leq 4\Bigl(\frac{(2^{\mu }\mu )^{p}q^{1-p}}{p}\widehat{A}_{*}^{p}\widehat{\s}_{1}^{1+p}+\widehat{\s}_{\mu }\Bigr)\label{eq:large:time:moment:2}
 \end{align}
for all $t\geq T$.
\end{lemma}

\begin{remark}
 In the language of dynamical systems, the previous lemma in particular yields that the semi-group $S(t)$ associated to~\eqref{eq:Smol:1} is dissipative (see~\cite[Definition~10.2]{Rob01}) on each space $\ell_{\mu }^{1}$ if $\mu >\max\{2-\alpha-\beta,1\}$ and thus has a global attractor (see~\cite[Theorem~10.5]{Rob01}).
\end{remark}

\begin{proof}[Proof of Lemma~\ref{Lem:moment:general:large:time}]
 We can proceed similarly as in the proof of Lemma~\ref{Lem:moment:general}. More precisely, Lemma~\ref{Lem:total:mass:1} guarantees the existence of $T_{1}>0$ such that $\mom_{1}(t)\leq 2\widehat{\s}_{1}$ for all $t\geq T_{1}$. Thus, together with~\eqref{eq:moment:proof:1} we find
 \begin{equation*}
 (\mom_{\mu }^{N})^{\frac{\mu +\gamma-1}{\mu -1}}\leq (2\widehat{\s}_{1})^{\frac{\gamma}{\mu -1}} \mom_{\mu +\gamma}^{N}.
\end{equation*}
Since $\widehat{\s}_{1}>0$ and $\mom_{1}^{N}(t)\leq 2\widehat{\s}_{1}$ uniformly in $t$ and $N$, we thus obtain from~\eqref{eq:higher:moment:3} that
 \begin{equation}\label{eq:moment:general:large:1}
  \frac{\dd}{\dt}\mom_{\mu }^{N}(t)+2^{-1-\frac{\gamma}{\mu -1}} R_{*}\widehat{\s}_{1}^{-\frac{\gamma}{\mu -1}}(\mom_{\mu }^{N})^{\frac{\mu +\gamma-1}{\mu -1}}\leq \frac{(2^{\mu }\mu )^{p}q^{1-p}}{p}R_{*}^{1-p}A_{*}^{p}\widehat{\s}_{1}^{1+p}+\s_{\mu } \quad \text{for all }t\geq T_{1}.
 \end{equation}
 If we apply Lemma~\ref{Lem:nonlinear:gronwall} with $t_{0}=T_{1}$ as well as
 \begin{equation}
   \Xi\vcc=\frac{(2^{\mu }\mu )^{p}q^{1-p}}{p}R_{*}^{1-p}A_{*}^{p}\widehat{\s}_{1}^{1+p}+\s_{\mu }\quad \text{and}\quad
   \Lambda\vcc=2^{-1-\frac{\gamma}{\mu -1}} R_{*}\widehat{\s}_{1}^{-\frac{\gamma}{\mu -1}}
 \end{equation}
 we obtain
 \begin{equation*}
  \mom_{\mu }^{N}(t)\leq \max\biggl\{\Bigl(\frac{2^{2+\rho_{\mu }}(2^{\mu }\mu )^{p}q^{1-p}}{p} \widehat{A}_{*}^{p}\widehat{\s}_{1}^{1+p+\rho_{\mu }}+2^{2+\rho_{\mu }}\widehat{\s}_{\mu }\widehat{\s}_{1}^{\rho_{\mu }}\Bigr)^{\frac{1}{1+\rho_{\mu }}},2\Bigl(\frac{4}{R_{*}\rho_{\mu }}\Bigr)^{\frac{1}{\rho_{\mu }}}\widehat{s}_{1}(t-T_{1})^{-\frac{1}{\rho_{\mu }}} \biggr\}
 \end{equation*}
 for all $t>T_{1}$ taking also~\eqref{eq:reduce:coefficients} and $\rho_{\mu }=\gamma/(\mu -1)$ into account. Taking $T_{2}>T_{1}$ such that 
 \begin{equation*}
  2\Bigl(\frac{4}{R_{*}\rho_{\mu }}\Bigr)^{\frac{1}{\rho_{\mu }}}\widehat{s}_{1}(t-T_{1})^{-\frac{1}{\rho_{\mu }}}\leq \Bigl(\frac{2^{2+\rho_{\mu }}(2^{\mu -1}\mu )^{p}q^{1-p}}{p} \widehat{A}_{*}^{p}\widehat{\s}_{1}^{1+p+\rho_{\mu }}+2^{2+\rho_{\mu }}\widehat{\s}_{\mu }\widehat{\s}_{1}^{\rho_{\mu }}\Bigr)^{\frac{1}{1+\rho_{\mu }}}
 \end{equation*}
 for all $t>T_{2}$ and taking the limit $N\to\infty$ the estimate~\eqref{eq:large:time:moment:1} follows.
 
 Similarly, using $\mom_{1}(t)\leq 2\widehat{\s}_{1}$ for all $t>T_{1}$ and $\mom_{\mu }^{N}(t)\leq \mom_{\mu +\gamma}^{N}$ since $\gamma>0$, we obtain from~\eqref{eq:higher:moment:3} that
 \begin{equation*}
  \frac{\dd}{\dt}\mom_{\mu }^{N}(t)+\frac{R_{*}}{2}\mom_{\mu }^{N}(t)\leq \frac{(2^{\mu }\mu )^{p}q^{1-p}}{p}R_{*}\widehat{A}_{*}^{p}\widehat{\s}_{1}^{1+p}+\s_{\mu }.
 \end{equation*}
 Thus, Grönwall's inequality directly yields
 \begin{equation}\label{eq:proof:large:moments:1}
  \mom_{\mu }^{N}(t)\leq \mom_{\mu }^{N}(T_{2})\ee^{-\frac{R_{*}}{2}(t-T_{2})}+2\Bigl(\frac{(2^{\mu }\mu )^{p}q^{1-p}}{p}\widehat{A}_{*}^{p}\widehat{\s}_{1}^{1+p}+\widehat{\s}_{\mu }\Bigr)\bigl(1-\ee^{-\frac{R_{*}}{2}(t-T_{2})}\bigr).
 \end{equation}
 We have already shown that
 \begin{equation*}
  \mom_{\mu }^{N}(T_{2})\leq 2\Bigl(\frac{2^{2+\rho_{\mu }}(2^{\mu }\mu )^{p}q^{1-p}}{p} \widehat{A}_{*}^{p}\widehat{\s}_{1}^{1+p+\rho_{\mu }}+2^{2+\rho_{\mu }}\widehat{\s}_{\mu }\widehat{\s}_{1}^{\rho_{\mu }}\Bigr)^{\frac{1}{1+\rho_{\mu }}}.
 \end{equation*}
 Thus, to conclude the proof of~\eqref{eq:large:time:moment:2}, it suffices to choose $T>T_{2}$ sufficiently large such that 
 \begin{multline*}
  2\Bigl(\frac{2^{2+\rho_{\mu }}(2^{\mu }\mu )^{p}q^{1-p}}{p} \widehat{A}_{*}^{p}\widehat{\s}_{1}^{1+p+\rho_{\mu }}+2^{2+\rho_{\mu }}\widehat{\s}_{\mu }\widehat{\s}_{1}^{\rho_{\mu }}\Bigr)^{\frac{1}{1+\rho_{\mu }}}\ee^{-\frac{R_{*}}{2}(t-T_{2})}\\*
  \leq 2\Bigl(\frac{(2^{\mu }\mu )^{p}q^{1-p}}{p}\widehat{A}_{*}^{p}\widehat{\s}_{1}^{1+p}+\widehat{\s}_{\mu }\Bigr) \quad \text{for all }t\geq T
 \end{multline*}
 and to take the limit $N\to\infty$ in~\eqref{eq:proof:large:moments:1}.
\end{proof}

\section{Existence of a solution}\label{Sec:existence}

To prove existence of solutions to~\eqref{eq:Smol:1} we follow an approach for discrete coagulation (-fragmentation) equations which has been used in similar form in previous works e.g.\@~\cite{BaC90,Cos95a}.

More precisely, we consider first a finite dimensional approximation of~\eqref{eq:Smol:1} which, for $N\in\N$ fixed, reads
\begin{equation}\label{eq:Smol:finite}
 \begin{aligned}
  \frac{\dd}{\dt}c_{k}^{N}&=\frac{1}{2}\sum_{\ell=1}^{k-1}a_{k-\ell,\ell}c_{k-\ell}^{N}c_{\ell}^{N}-c_{k}^{N}\sum_{\ell=1}^{N-k}a_{k,\ell}c_{\ell}^{N}+s_{k}-r_{k}c_{k}^{N} &&\text{if }k\leq N,\\
  c_{k}^{N}&=0 && \text{else}.
 \end{aligned}
\end{equation}
We note that the sequence $c^{N}=(c_{k}^{N})_{k\in\N}$ is a solution to~\eqref{eq:Smol:1} where the coefficients $a_{k,\ell}$, $s_{k}$ and the initial condition $c^{\text{in}}=(c_{k}^{\text{in}})_{k\in\N}$ are replaced by
\begin{equation*}
 a_{k,\ell}^{N}=\begin{cases}
                 a_{k,\ell} & \text{if }k+\ell\leq N\\ 
                 0 &\text{else}
                \end{cases},
\quad
s_{k}^{N}=\begin{cases}
           s_{k} & \text{if  }k\leq N\\
           0 & \text{else}
          \end{cases}
\quad  \text{and}\quad
c_{k}^{N,\text{in}}=\begin{cases}
                  c^{\text{in}}_{k} & \text{if }k\leq N\\
                  0 &\text{else}.
                 \end{cases}
\end{equation*}
In particular, since $a_{k,\ell}^{N}$ and $s_{k}^{N}$ satisfy the assumptions~\cref{eq:Ass:a,eq:Ass:s} the moment estimates derived in Section~\ref{Sec:moment:estimates} still hold for $c^{N}$.

\subsection{Existence of a solution for the truncated system}

The following proposition states the existence of a unique classical global solution for~\eqref{eq:Smol:finite}

\begin{proposition}\label{Prop:existence:truncated}
 For each fixed $N\in\N$ the system~\eqref{eq:Smol:finite} has a unique solution $c^{N}=(c_{k}^{N})_{k\in\N}$ such that $c_{k}^{N}\in C^{1}([0,\infty),\R_{\geq 0})$ for each $k\in\N$. 
\end{proposition}

\begin{proof}
 The proof of this statement follows from classical arguments of the theory of ordinary differential equations. For the sake of completeness, we outline the main steps. To simplify the notation we define functions $f_{k}\colon \R^{N}\to \R^{N}$ for $k=1,\ldots N$ through 
 \begin{equation*}
  f_{k}(x_{1},\ldots, x_{N})=\frac{1}{2}\sum_{\ell=1}^{k-1}a_{k-\ell,\ell}x_{k-\ell}x_{\ell}-x_{k}\sum_{\ell=1}^{N-k}a_{k,\ell}x_{\ell}+s_{k}-r_{k}x_{k}.
 \end{equation*}
 Thus, \eqref{eq:Smol:finite} can be rewritten as
 \begin{equation*}
  \frac{\dd}{\dd{t}}c_{k}^{N}=f_{k}(c_{1}^{N},\ldots,c_{N}^{N})\quad \text{for }k=1,\ldots, N.
 \end{equation*}
 As polynomials, the functions $f_{k}$ are in particular locally Lipschitz continuous. Thus, due to the Picard-Lindelöf theorem there exists a unique solution $c^{N}=(c_{k}^{N})_{k\in\N}$ on a maximal time interval $[0,T_{*})$, i.e.\@ $c_{k}^{N}\in C^{1}([0,T_{*}))$ for all $k=1,\ldots N$. 
 
 On the other hand, the function $f=(f_{k})_{k}$ is \emph{quasi-positive} in the notion of~\cite{PrW10} which precisely means that
 \begin{equation*}
  f_{k}(x_{1},\ldots, x_{k-1}, 0, x_{k+1},\ldots, x_{N})\geq 0\quad \text{if } x_{j}\geq 0 \text{ for all } j\in\{1,\ldots, N\}\setminus\{k\}. 
 \end{equation*}
 The validity of this property is easily checked since $a_{k,\ell}, s_{k}\geq 0$ by assumption. Thus, since $c_{k}^{\text{in},N}\geq 0$ for all $k\in\N$ it follows from~\cite[Satz~4.2.2.]{PrW10} that $c_{k}^{N}\geq 0$ on $[0,T_{*})$ for all $k=1,\ldots, N$.
 
 To show $T_{*}=\infty$ we rely on the moment estimates from Section~\ref{Sec:moment:estimates}. In fact, Lemma~\ref{Lem:total:mass:1} implies 
 \begin{equation*}
  \sup_{t\in[0,T_{*})} \sum_{k=1}^{N}kc_{k}^{N}(t)\leq \max\bigl\{\mom_{1}^{\text{in}},\widehat{\s}_{1}\bigr\}.
 \end{equation*}
 Thus, $c_{k}^{N}$ cannot blow up on the interval $[0,T_{*})$ which implies that $T_{*}=\infty$ (see~\cite[Satz~2.3.2]{PrW10}).
\end{proof}

\subsection{Existence of a global solution to~\eqref{eq:Smol:1}}

The general goal will be to pass to the limit $N\to\infty$ in the finite system~\eqref{eq:Smol:finite}. For this, we rely on the moment estimates from Section~\ref{Sec:moment:estimates} in order to obtain compactness. This approach is by now classical for coagulation (-fragmentation) equations and we follow here mainly~\cite{Lau02,Cos95a}. 

\begin{lemma}\label{Lem:est:der:1}
 For each fixed $T>0$ and $k\in\N$ there exists a constant $\mathscr{C}$ which depends on $T$ and $k$ but which is independent of $N$ such that
 \begin{equation*}
  \norm*{\frac{\dd}{\dd{t}}c_{k}^{N}(\cdot)}_{L^{1}(0,T)}\leq \mathscr{C}
 \end{equation*}
 for all $N\geq k$ where $c^{N}=(c_{k}^{N})_{k\in\N}$ is the solution so~\eqref{eq:Smol:finite} provided by Proposition~\ref{Prop:existence:truncated}.
\end{lemma}

\begin{proof}
 Due to the assumptions \cref{eq:Ass:a,eq:Ass:r} we have
 \begin{equation*}
  \abs*{\frac{\dd}{\dd{t}}c_{k}^{N}(t)}\leq \frac{A_{*}}{2}\sum_{\ell=1}^{k-1}\bigl((k-\ell)^{\alpha}\ell^{\beta}+(k-\ell)^{\beta}\ell^{\alpha}\bigr)c_{k-\ell}^{N}c_{\ell}^{N}+A_{*}c_{k}^{N}\sum_{\ell=1}^{N-k}(k^{\alpha}\ell^{\beta}+k^{\beta}\ell^{\alpha})c_{\ell}^{N}+s_{k}+R_{*}k^{\gamma}c_{k}^{N}.
 \end{equation*}
 Using the trivial estimates $k^{\mu }c_{k}^{N}(t)\leq \norm{c^{N}(t)}_{\ell_{\mu }^{1}}$ and $k^{\gamma}c_{k}^{N}(t)\leq k^{\gamma-1}\norm{c^{N}(t)}_{\ell_{1}^{1}}$ we further find
 \begin{equation*}
  \abs*{\frac{\dd}{\dd{t}}c_{k}^{N}(t)}\leq 3A_{*}\norm{c^{N}(t)}_{\ell_{\alpha}^{1}}\norm{c^{N}(t)}_{\ell_{\beta}^{1}}+s_{k}+R_{*}k^{\gamma-1}\norm{c^{N}(t)}_{\ell_{1}^{1}}.
 \end{equation*}
 Since $\alpha\leq \beta\leq 1$ we recall from Remark~\ref{Rem:moments:monotonicity} that $\norm{c^{N}(t)}_{\ell_{\alpha}^{1}}\leq \norm{c^{N}(t)}_{\ell_{\beta}^{1}}=\norm{c^{N}(t)}_{\ell_{1}^{1}}= \mom_{1}^{N}(t)$. Together with $s_{k}\leq \s_{1}$ this further yields
 \begin{equation}\label{eq:der:1:1}
  \abs*{\frac{\dd}{\dd{t}}c_{k}^{N}(t)}\leq 3A_{*}\bigl(\mom_{1}^{N}(t)\bigr)^{2}+\s_{1}+R_{*}k^{\gamma-1}\mom_{1}^{N}(t).
 \end{equation}
 Recalling once more from Lemma~\ref{Lem:total:mass:1} that $\mom_{1}^{N}(t)\leq \max\{\mom_{1}^{\text{in}},\widehat{\s}_{1}\}$ we immediately conclude
 \begin{equation*}
     \norm*{\frac{\dd}{\dd{t}}c_{k}^{N}(\cdot)}_{L^{1}(0,T)}\leq \Bigl(3A_{*}\max\{\mom_{1}^{\text{in}},\widehat{\s}_{1}\}^{2}+\s_{1}+R_{*}k^{\gamma-1}\max\{\mom_{1}^{\text{in}},\widehat{\s}_{1}\}\Bigr)T
 \end{equation*}
 which finishes the proof.
\end{proof}

The next proposition provides a certain stability of the right-hand side of~\eqref{eq:Smol:1} which will be the key in passing to the limit $N\to\infty$ in~\eqref{eq:Smol:finite}.

\begin{proposition}\label{Prop:convergence:operator}
 Let $(c^{n})_{n\in\N}$ and $(s^{n})_{n\in\N}$ be sequences with $c^{n}=(c^{n}_{k})_{k\in\N}\in L^{\infty}([0,\infty),\ell_{1}^{1})$ for each $n\in\N$ and assume that there exist $c=(c_{k})_{k\in\N}$ and $s=(s_{k})_{k\in\N}$ such that $s_{k}^{n}\to s_{k}$ as $n\to\infty$ as well as
 \begin{equation*}
  c_{k}^{n}(t)\longrightarrow c_{k}(t) \quad \text{as }n\longrightarrow \infty \text{ uniformly on compact subsets of }[0,\infty).
 \end{equation*}
 Let $a_{k,\ell}^{n}$ satisfy~\eqref{eq:Ass:a} uniformly with respect to $n\in\N$ and let $a_{k,\ell}^{n}\to a_{k,\ell}$ as $n\to\infty$ for each $k, \ell\in\N$. Assume further that there exists a constant $M_{1}>0$ such that 
 \begin{equation*}
  \mom_{1}^{n}(s)\vcc=\sum_{k=1}^{\infty}k c_{k}^{n}(s)\leq M_{1}\quad \text{for all } s\in[0,\infty) \quad \text{ and all }n\in\N.
 \end{equation*}
 Finally for some $\mu >1$ let $(c^{n})_{n\in\N}$ satisfy
 \begin{equation*}
  \mom_{\mu }^{n}(s)\vcc=\sum_{k=1}^{\infty}k^{\mu }c_{k}^{n}(s)\leq M_{\mu }(1+s^{-\rho})\quad \text{for all } s\in(0,\infty)\text{ and }n\in\N
 \end{equation*}
 with constants $M_{\mu }>0$ and $\rho\in[0,1)$. Then, we have for each $k\in\N$ fixed that
 \begin{multline}\label{eq:conv:op:1}
  \lim_{n\to\infty}\biggl(\frac{1}{2}\sum_{\ell=1}^{k-1}a^{n}_{k-\ell,\ell}c_{k-\ell}^{n}c_{\ell}^{n}-c_{k}^{n}\sum_{\ell=1}^{\infty}a^{n}_{k,\ell}c_{\ell}^{n}+s_{k}^{n}-r_{k}c_{k}^{n}\biggr)\\*
  =\frac{1}{2}\sum_{\ell=1}^{k-1}a_{k-\ell,\ell}c_{k-\ell}c_{\ell}-c_{k}\sum_{\ell=1}^{\infty}a_{k,\ell}c_{\ell}+s_{k}-r_{k}c_{k}
 \end{multline}
 uniformly on each compact subset of $(0,\infty)$. Moreover, we have for each $t>0$ that
 \begin{multline}\label{eq:conv:op:2}
  \lim_{n\to\infty}\biggl(\int_{0}^{t}\frac{1}{2}\sum_{\ell=1}^{k-1}a^{n}_{k-\ell,\ell}c_{k-\ell}^{n}(s)c_{\ell}^{n}(s)-c_{k}^{n}(s)\sum_{\ell=1}^{\infty}a^{n}_{k,\ell}c_{\ell}^{n}(s)+s_{k}^{n}-r_{k}c_{k}^{n}(s)\ds\biggr)\\*
  =\int_{0}^{t}\frac{1}{2}\sum_{\ell=1}^{k-1}a_{k-\ell,\ell}c_{k-\ell}(s)c_{\ell}(s)-c_{k}(s)\sum_{\ell=1}^{\infty}a_{k,\ell}c_{\ell}(s)+s_{k}-r_{k}c_{k}(s)\ds.
 \end{multline}
\end{proposition}

\begin{proof}
 We first note that by means of Fatou's Lemma, we immediately conclude from the assumptions of the proposition that
 \begin{equation*}
  \mom_{1}(s)\vcc=\sum_{k=1}^{\infty}kc_{k}(s)\leq M_{1}\quad \text{and}\quad \mom_{\mu }(s)\vcc=\sum_{k=1}^{\infty}k^{\mu }c_{k}(s)\leq M_{\mu }(1+s^{-\rho}).
 \end{equation*}
 We will first show the claimed convergence in~\eqref{eq:conv:op:1} from which~\eqref{eq:conv:op:2} will then easily follow. 
 
 To prove~\eqref{eq:conv:op:1}, we consider the different terms separately and note first that we already have $s_{k}^{n}\to s_{k}$ due to the assumptions. Moreover, since $c_{\ell}^{n}\to c_{\ell}$ uniformly on compact subsets of $[0,\infty)$ as $n\to\infty$, we immediately get
 \begin{equation}\label{eq:general:conv:1}
  \frac{1}{2}\sum_{\ell=1}^{k-1}a^{n}_{k-\ell,\ell}c_{k-\ell}^{n}c_{\ell}^{n}\longrightarrow \frac{1}{2}\sum_{\ell=1}^{k-1}a_{k-\ell,\ell}c_{k-\ell}c_{\ell}\quad \text{and}\quad r_{k}c_{k}^{n}\longrightarrow r_{k}c_{k} \quad \text{as }n\longrightarrow \infty
 \end{equation}
 locally uniformly on $[0,\infty)$. Thus, it remains to estimate the difference $c_{k}^{n}\sum_{\ell=1}^{\infty}a^{n}_{k,\ell}c_{\ell}^{n}-c_{k}\sum_{\ell=1}^{\infty}a_{k,\ell}c_{\ell}$ which we rewrite for some $Z\in \N$ to get
 \begin{multline}\label{eq:splitting:nonlinear:term}
  \abs*{c_{k}^{n}\sum_{\ell=1}^{\infty}a_{k,\ell}^{n}c_{\ell}^{n}-c_{k}\sum_{\ell=1}^{\infty}a_{k,\ell}c_{\ell}}\\*
  \leq \abs{c_{k}^{n}-c_{k}}\sum_{\ell=1}^{\infty}a_{k,\ell}^{n}c_{\ell}^{n}+c_{k}\sum_{\ell=1}^{Z-1}\Bigl(\abs*{a_{k,\ell}^{n}-a_{k,\ell}}c_{\ell}^{n}+a_{k,\ell}\abs{c_{\ell}^{n}-c_{\ell}}\Bigr)+c_{k}\sum_{\ell=Z}^{\infty}\Bigl(a_{k,\ell}^{n}c_{\ell}^{n}+a_{k,\ell}c_{\ell}\Bigr).
 \end{multline}
 For the first expression on the right-hand side we obtain together with~\eqref{eq:Ass:a} and $k^{\alpha}\ell^{\beta}+k^{\beta}\ell^{\alpha}\leq 2(k\ell)^{\beta}$ as well as $\ell^{\beta}\leq \ell$ for all $k,\ell\in\N$ since $\alpha\leq \beta$ that
 \begin{multline}\label{eq:general:conv:2}
  \abs{c_{k}^{n}(s)-c_{k}(s)}\sum_{\ell=1}^{\infty}a_{k,\ell}^{n}c_{\ell}^{n}(s)\leq A_{*}\abs*{c_{k}^{n}(s)-c_{k}(s)}\sum_{\ell=1}^{\infty}(k^{\alpha}\ell^{\beta}+k^{\beta}\ell^{\alpha})c_{\ell}^{n}(s)\\*
  \leq 2A_{*}k^{\beta}\abs{c_{k}^{n}(s)-c_{k}(s)}\sum_{\ell=1}^{\infty}\ell c_{\ell}^{n}(s)\leq 2A_{*}M_{1}k^{\beta}\abs{c_{k}^{n}(s)-c_{k}(s)}\longrightarrow 0 
 \end{multline}
 as $n\to\infty$ locally uniformly on $[0,\infty)$. For the second expression on the right-hand side of~\eqref{eq:splitting:nonlinear:term} we find in the same way as in~\eqref{eq:general:conv:1} that
 \begin{equation}\label{eq:general:conv:3}
  \lim_{n\to\infty}\biggl(c_{k}(s)\sum_{\ell=1}^{Z-1}\Bigl(\abs*{a_{k,\ell}^{n}-a_{k,\ell}}c_{\ell}^{n}(s)+a_{k,\ell}\abs{c_{\ell}^{n}(s)-c_{\ell}(s)}\Bigr)\biggr)=0 \quad \text{locally uniformly on }[0,\infty).
 \end{equation}
 Finally, for the third term on the right-hand side of~\eqref{eq:splitting:nonlinear:term} we note that $\mu $, as specified in the assumptions, satisfies $\mu >\beta$ since $\beta\in[0,1]$. Thus, using additionally that $a_{k,\ell},a_{k,\ell}^{n}\leq 2A_{*} k^{\beta}\ell^{\beta}$ as before and $k^{\beta}c_{k}(s)\leq M_{1}$ we get
 \begin{multline}\label{eq:general:conv:critical}
  c_{k}(s)\sum_{\ell=Z}^{\infty}\Bigl(a_{k,\ell}^{n}c_{\ell}^{n}(s)+a_{k,\ell}c_{\ell}(s)\Bigr)\leq 2A_{*}k^{\beta}c_{k}(s)\sum_{\ell=Z}^{\infty}\ell^{\beta}\bigl(c_{\ell}^{n}(s)+c_{\ell}(s)\bigr)\\*
  \leq 2A_{*}M_{1}\sum_{\ell=Z}^{\infty}\ell^{\beta-\mu }\ell^{\mu }\bigl(c_{\ell}^{n}(s)+c_{\ell}(s)\bigr)\\*
  \leq 2A_{*}M_{1}Z^{\beta-\mu }\bigl(\mom_{\mu }^{n}(s)+\mom_{\mu }(s)\bigr)\leq 4A_{*}M_{1}M_{\mu }(1+s^{-\rho})Z^{\beta-\mu }.
 \end{multline}
Since $\mu >\beta$ and $\rho\in[0,1)$, the right-hand side converges to zero locally uniformly on $(0,\infty)$ and independently of $n$ as $Z\to\infty$. Thus, the relation~\eqref{eq:conv:op:1} follows upon collecting \cref{eq:general:conv:1,eq:general:conv:2,eq:general:conv:3} and taking first the limit $n\to\infty$ and then $Z\to\infty$. Furthermore, since the estimates \cref{eq:general:conv:1,eq:general:conv:2,eq:general:conv:3} are uniform with respect to $s$ and since $\int_{0}^{t}(1+s^{-\rho})\ds\leq (t+t^{1-\rho}/(1-\rho))$ the convergence~\eqref{eq:conv:op:2} is a direct consequence.
\end{proof}

We are now prepared to prove the existence of global solutions to~\eqref{eq:Smol:1}.

\begin{proof}[Proof of Theorem~\ref{Thm:existence:evolution}]
 Let $(c^{N})_{N\in\N}$ be the sequence of solutions for the finite systems provided by Proposition~\ref{Prop:existence:truncated}. According to Lemma~\ref{Lem:total:mass:1} we have $\mom_{1}^{N}(t)\leq \max\{\mom_{1}^{\text{in}},\widehat{\s}_{1}\}$ such that Lemma~\ref{Lem:est:der:1} ensures that for fixed $T>0$ and $k\in\N$ the sequence $(c_{k}^{N})_{N\geq k}$ is uniformly bounded in $L^{\infty}(0,T)\cap W^{1,1}(0,T)$. Thus, an Arzela-Ascoli type argument yields that there exists a sub-sequence of $c^{N}=(c_{k}^{N})_{N\in\N}$ (which we will not relabel) and a sequence $c=(c_{k})_{k\in\N}$ which is of locally bounded variation such that
 \begin{equation}\label{eq:existence:1}
  c_{k}^{N}(t)\longrightarrow c_{k}(t) \quad \text{as } N\longrightarrow \infty
 \end{equation}
 uniformly on each compact subset of $[0,\infty)$. Due to Fatou's Lemma and \cref{Lem:total:mass:1,Lem:moment:general} we have additionally that
 \begin{equation}\label{eq:limit:moments}
  \mom_{1}(t)\leq \max\bigl\{\mom_{1}^{\text{in}},\widehat{\s}_{1}\bigr\}\quad \text{and}\quad \mom_{\mu }(t)\leq C_{\mu }(1+t^{-\rho}) \quad\text{for some }\mu >0 \text{ and }\rho\in[0,1).
 \end{equation}
 Moreover, equation~\eqref{eq:Smol:finite} can be rewritten as
 \begin{equation*}
  c_{k}^{N}(t)=c_{k}^{N,\text{in}}+\int_{0}^{t}\biggl(\frac{1}{2}\sum_{\ell=1}^{k-1}a_{k-\ell,\ell}^{N}c_{k-\ell}^{N}(s)c_{\ell}^{N}(s)-c_{k}^{N}(s)\sum_{\ell=1}^{\infty}a_{k,\ell}^{N}c_{\ell}^{N}(s)+s_{k}-r_{k}c_{k}^{N}(s)\biggr)\ds.
 \end{equation*}
 Thus, taking the limit $N\to\infty$, Proposition~\ref{Prop:convergence:operator} yields
 \begin{equation}\label{eq:existence:integrated}
  c_{k}(t)=c_{k}^{\text{in}}+\int_{0}^{t}\biggl(\frac{1}{2}\sum_{\ell=1}^{k-1}a_{k-\ell,\ell}c_{k-\ell}(s)c_{\ell}(s)-c_{k}(s)\sum_{\ell=1}^{\infty}a_{k,\ell}c_{\ell}(s)+s_{k}-r_{k}c_{k}(s)\biggr)\ds.
 \end{equation}
 To conclude, we note that from~\eqref{eq:existence:1} together with~\eqref{eq:limit:moments} we obtain $c\in L^{\infty}([0,T),\ell_{1}^{1})\cap C^{1}((0,T),\ell_{\mu}^{1})$ and that $c_{k}\colon [0,\infty)\to [0,\infty)$ is continuous with $c(0)=c^{\text{in}}$. Differentiating in~\eqref{eq:existence:integrated} finally yields that $c$ solves~\eqref{eq:Smol:1}.
\end{proof}

\section{A functional inequality}\label{Sec:func:inequality}

In this section we are going to derive a functional inequality similar to~\cite{FoM04}, which will be the key in the proofs of \cref{Thm:existence:stat:sol,Thm:convergence}. As a preliminary step we prove the following lemma.

\begin{lemma}\label{Lem:techhnical:estimate:1}
 For each $\mu\geq 1$ we have
 \begin{equation*}
  (k+\ell)^{\mu}-k^{\mu}+\ell^{\mu}\leq 2^{\max\{\mu-2,0\}}\max\{\mu,\mu(\mu-1)\}\ell^{\max\{1,\mu-1\}}k^{\mu-1}+2\ell^{\mu}
 \end{equation*}
 for all $k,\ell\in\N$.
\end{lemma}

\begin{proof}
 It turns out to be convenient to estimate the expression
 \begin{equation*}
  \frac{(k+\ell)^{\mu}-k^{\mu}-\ell^{\mu}}{k^{\mu-1}}=\ell\frac{k}{\ell}\biggl(\Bigl(1+\frac{\ell}{k}\Bigr)^{\mu}-1-\Bigl(\frac{\ell}{k}\Bigr)^{\mu}\biggr).
 \end{equation*}
 Thus, defining $z=\ell/k$ the right-hand side reads $\ell h_{\mu}(z)$ with 
 \begin{equation*}
  h_{\mu}(z)=\frac{(1+z)^{\mu}-1-z^{\mu}}{z}.
 \end{equation*}
 The expression $h_{\mu}(z)$ can be rewritten as
 \begin{equation}\label{eq:tech:est:2}
  h_{\mu}(z)=\frac{(1+z)^{\mu}-1-z^{\mu}}{z}=\frac{1}{z}\biggl(\int_{0}^{z}\del_{x}(1+x)^{\mu}\dx-\int_{0}^{z}\del_{x}x^{\mu}\dx\biggr)=\frac{\mu}{z}\int_{0}^{z}(1+x)^{\mu-1}-x^{\mu-1}\dx.
 \end{equation}
 Now, we have to treat the two cases $\mu\leq 2$ and $\mu>2$ separately. In the first one, i.e.\@ for $\mu\leq 2$ we note that the map $x\mapsto x^{\mu-1}$ is Hölder continuous with exponent $\mu-1$ and constant one, i.e.\@ we have $(1+x)^{\mu-1}-x^{\mu-1}\leq 1$ for all $x\geq 0$. Thus, we conclude directly from~\eqref{eq:tech:est:2} that
 \begin{equation*}
  h_{\mu}(z)\leq \mu \qquad \text{for all }z\geq 0 \text{ if }\mu\leq 2.
 \end{equation*}
 Recalling $(k+\ell)^{\mu}-k^{\mu}-\ell^{\mu}=\ell k^{\mu-1} h_{\mu}(\ell/k)$ yields the estimate
 \begin{equation}\label{eq:tech:est:3}
  (k+\ell)^{\mu}-k^{\mu}-\ell^{\mu}\leq \mu \ell k^{\mu-1} \qquad \text{for all }k,\ell\in\N\text{ if }\mu\leq 2.
 \end{equation}
On the other hand, if $\mu>2$, we rewrite the right-hand side of~\eqref{eq:tech:est:2} further to obtain
 \begin{equation*}
  h_{\mu}(z)=\frac{\mu}{z}\int_{0}^{z}\int_{x}^{x+1}\del_{y}y^{\mu-1}\dy\dx=\frac{\mu(\mu-1)}{z}\int_{0}^{z}\int_{x}^{x+1}y^{\mu-2}\dy\dx.
 \end{equation*}
 Since for $\mu>2$ the map $y\mapsto y^{\mu-2}$ is increasing, we can estimate the right-hand side to get
 \begin{equation*}
  h_{\mu}(z)=\frac{\mu(\mu-1)}{z}\int_{0}^{z}\int_{x}^{x+1}y^{\mu-2}\dy\dx\leq \frac{\mu(\mu-1)}{z}\int_{0}^{z}(x+1)^{\mu-2}\dx\leq \mu(\mu-1)(z+1)^{\mu-2}.
 \end{equation*}
 Again, we use $(k+\ell)^{\mu}-k^{\mu}-\ell^{\mu}=\ell k^{\mu-1} h_{\mu}(\ell/k)$ to find for all $k,\ell\in\N$ that
 \begin{equation*}
  (k+\ell)^{\mu}-k^{\mu}-\ell^{\mu}\leq \mu(\mu-1)\ell k^{\mu-1}\Bigl(\frac{\ell}{k}+1\Bigr)^{\mu-2}=\mu(\mu-1)k\ell (k+\ell)^{\mu-2} \qquad \text{if }\mu>2.
 \end{equation*}
To estimate the right-hand side further, we note that $k+\ell\leq 2k\ell$ for all $k,\ell\in\N$. This further implies
\begin{equation*}
 (k+\ell)^{\mu}-k^{\mu}-\ell^{\mu}\leq 2^{\mu-2}\mu(\mu-1)k^{\mu-1}\ell^{\mu-1} \qquad \text{if }\mu>2.
\end{equation*}
Combining this estimate with~\eqref{eq:tech:est:3} we obtain
\begin{equation*}
 (k+\ell)^{\mu}-k^{\mu}-\ell^{\mu}\leq 2^{\max\{\mu-2,0\}}\max\{\mu,\mu(\mu-1)\}\ell^{\max\{1,\mu-1\}}k^{\mu-1} \quad \text{for all }k,\ell\in\N \text{ and all }\mu\geq 1.
\end{equation*}
From this the claim immediately follows.
\end{proof}

We can now establish the following functional inequality satisfied by the difference of two solutions to~\eqref{eq:Smol:1}.

\begin{lemma}\label{Lem:functional:inequality:1:b}
 Let $c=(c_{k})_{k\in\N}$ and $d=(d_{k})_{k\in\N}$ be solutions to~\eqref{eq:Smol:1}. For each $\mu\geq 1$ we have
 \begin{equation*}
  \frac{\dd}{\dd{t}}\sum_{k=1}^{\infty}k^{\mu}\abs{c_{k}-d_{k}}\leq \biggl(2A_{*}(C_{\mu}+2)\sum_{\ell=1}^{\infty}\ell^{\mu+\beta}(c_{\ell}+d_{\ell})-R_{*}\biggr)\sum_{k=1}^{\infty}k^{\mu}\abs{c_{k}-d_{k}}\quad \text{for all }t>0
 \end{equation*}
 with $C_{\mu}\vcc=2^{\max\{\mu-2,0\}}\max\{\mu,\mu(\mu-1)\}$.
\end{lemma}

\begin{proof}
  We take the difference of the relations~\eqref{eq:Smol:weak} for $c_{k}$ and $d_{k}$ which, after some rearrangement, reads
 \begin{equation*}
  \frac{\dd}{\dd{t}}\sum_{k=1}^{\infty}(c_{k}-d_{k})\varphi_{k}=\frac{1}{2}\sum_{k=1}^{\infty}\sum_{\ell=1}^{\infty}a_{k,\ell}\bigl[(c_{k}-d_{k})c_{\ell}+(c_{\ell}-d_{\ell})d_{k}\bigr]\bigl[\varphi_{k+\ell}-\varphi_{k}-\varphi_{\ell}\bigr]-\sum_{k=1}^{\infty}r_{k}(c_{k}-d_{k})\varphi_{k}.
 \end{equation*}
Choosing $\varphi_{k}=k^{\mu}\sgn(c_{k}-d_{k})$ leads to the estimate
\begin{multline*}
 \frac{\dd}{\dd{t}}\sum_{k=1}^{\infty}k^{\mu}\abs{c_{k}-d_{k}}\leq \frac{1}{2}\sum_{k=1}^{\infty}\sum_{\ell=1}^{\infty}a_{k,\ell}c_{\ell}\abs{c_{k}-d_{k}}\bigl[(k+\ell)^{\mu}-k^{\mu}+\ell^{\mu}\bigr]\\*
 +\frac{1}{2}\sum_{k=1}^{\infty}\sum_{\ell=1}^{\infty}a_{k,\ell}d_{k}\abs{c_{\ell}-d_{\ell}}\bigl[(k+\ell)^{\mu}+k^{\mu}-\ell^{\mu}\bigr]-\sum_{k=1}^{\infty}k^{\mu}r_{k}\abs{c_{k}-d_{k}}.
\end{multline*}
Due to the symmetry of the kernel $a_{k,\ell}$, the right-hand side can be simplified yielding
\begin{equation}\label{eq:functional:ineq:1:b}
 \frac{\dd}{\dd{t}}\sum_{k=1}^{\infty}k^{\mu}\abs{c_{k}-d_{k}}\leq \frac{1}{2}\sum_{k=1}^{\infty}\sum_{\ell=1}^{\infty}a_{k,\ell}\bigl(c_{\ell}+d_{\ell}\bigr)\abs{c_{k}-d_{k}}\bigl[(k+\ell)^{\mu}-k^{\mu}+\ell^{\mu}\bigr]-\sum_{k=1}^{\infty}k^{\mu}r_{k}\abs{c_{k}-d_{k}}.
\end{equation}
Due to Lemma~\ref{Lem:techhnical:estimate:1} we have
\begin{equation*}
 (k+\ell)^{\mu}-k^{\mu}+\ell^{\mu}\leq C_{\mu}\ell^{\max\{1,\mu-1\}}k^{\mu-1}+2\ell^{\mu}.
\end{equation*}
Together with~\eqref{eq:Ass:a} and $\alpha\leq \beta\leq 1\leq \mu$ this leads to the estimate
\begin{align*}
 &\phantom{{}\leq{}} a_{k,\ell}\bigl[(k+\ell)^{\mu}-k^{\mu}+\ell^{\mu}\bigr]\\
 &\leq C_{\mu}A_{*}\bigl[k^{\mu-1+\alpha}\ell^{\max\{1,\mu-1\}+\beta}+k^{\mu-1+\beta}\ell^{\max\{1,\mu-1\}+\alpha}\bigr]+2A_{*}\bigl[k^{\alpha}\ell^{\mu+\beta}+k^{\beta}\ell^{\mu+\alpha}\bigr]\\
 &\leq 2A_{*}(C_{\mu}+2)k^{\mu}\ell^{\mu+\beta}.
\end{align*}
Plugging this into~\eqref{eq:functional:ineq:1:b} and recalling from~\eqref{eq:Ass:r} that $r_{k}\geq R_{*}k^{\gamma}\geq R_{*}$ it follows
\begin{equation*}
 \frac{\dd}{\dd{t}}\sum_{k=1}^{\infty}k^{\mu}\abs{c_{k}-d_{k}}\leq \biggl(2A_{*}(C_{\mu}+2)\sum_{\ell=1}^{\infty}\ell^{\mu+\beta}(c_{\ell}+d_{\ell})-R_{*}\biggr)\sum_{k=1}^{\infty}k^{\mu}\abs{c_{k}-d_{k}}.
\end{equation*}
This concludes the proof.
\end{proof}

As a preparation for the proof of Theorem~\ref{Thm:convergence} in Section~\ref{Sec:convergence}, we provide the following lemma, which is a corollary of Lemma~\ref{Lem:functional:inequality:1:b}.

\begin{lemma}\label{Lem:functional:inequality:2}
 Let $\mu\geq 1$ such that $\mu+\beta>\max\{2-\alpha-\beta,1\}$, let $p$ and $q$ be as in~\eqref{eq:Hoelder:exp:1} and let $\rho_{\mu+\beta}=\gamma/(\mu+\beta-1)$. Assume further that either
 \begin{align*}
  \kappa_{2}&\vcc=R_{*}-16C_{\mu}A_{*}\Bigl(\frac{(2^{\mu+\beta}(\mu+\beta))^{p}q^{1-p}}{p}\widehat{A}_{*}^{p}\widehat{\s}_{1}^{1+p}+\widehat{\s}_{\mu+\beta}\Bigr)>0
\shortintertext{or}
 \kappa_{1}&\vcc=R_{*}-8C_{\mu}A_{*}\Bigl(\frac{2^{2+\rho_{\mu+\beta}}(2^{\mu+\beta}(\mu+\beta))^{p}q^{1-p}}{p} \widehat{A}_{*}^{p}\widehat{\s}_{1}^{1+p+\rho_{\mu+\beta}}+2^{2+\rho_{\mu+\beta}}\widehat{\s}_{\mu+\beta}\widehat{\s}_{1}^{\rho_{\mu+\beta}}\Bigr)^{\frac{1}{1+\rho_{\mu+\beta}}}>0.
\end{align*}
 Then, $\kappa=\max\{\kappa_{1},\kappa_{2}\}>0$ and for each pair $c=(c_{k})_{k\in\N}$ and $d=(d_{k})_{k\in\N}$ of solutions  to~\eqref{eq:Smol:1} there exists a time $T_{*}>0$ such that
 \begin{equation}\label{eq:functional:inequality}
  \frac{\dd}{\dd{t}}\sum_{k=1}^{\infty}k^{\mu}\abs{c_{k}-d_{k}}\leq -\kappa\sum_{k=1}^{\infty}k^{\mu}\abs{c_{k}-d_{k}}\quad \text{for all }t\geq T_{*}.
 \end{equation}
\end{lemma}

\begin{proof}
According to Lemma~\ref{Lem:functional:inequality:1:b} it suffices to show that there exists $T_{*}>0$ such that
\begin{equation*}
 2A_{*}(2^{\max\{\mu-2,0\}}\max\{\mu,\mu(\mu-1)\}+2)\sum_{\ell=1}^{\infty}\ell^{\mu+\beta}(c_{\ell}+d_{\ell})-R_{*}\leq -\kappa \qquad \text{for all }t\geq T_{*}.
\end{equation*}
By means of Lemma~\ref{Lem:moment:general:large:time} we can estimate the left-hand side for sufficiently large $T_{*}>0$ to get
\begin{multline*}
 2C_{\mu}A_{*}\sum_{\ell=1}^{\infty}\ell^{\mu+\beta}(c_{\ell}+d_{\ell})-R_{*}\leq 4C_{\mu}A_{*}\min\biggl\{4\Bigl(\frac{(2^{\mu+\beta}(\mu+\beta))^{p}q^{1-p}}{p}\widehat{A}_{*}^{p}\widehat{\s}_{1}^{1+p}+\widehat{\s}_{\mu+\beta}\Bigr),\\*
 2\Bigl(\frac{2^{2+\rho_{\mu+\beta}}(2^{\mu+\beta}(\mu+\beta))^{p}q^{1-p}}{p} \widehat{A}_{*}^{p}\widehat{\s}_{1}^{1+p+\rho_{\mu+\beta}}+2^{2+\rho_{\mu+\beta}}\widehat{\s}_{\mu+\beta}\widehat{\s}_{1}^{\rho_{\mu+\beta}}\Bigr)^{\frac{1}{1+\rho_{\mu+\beta}}}\biggr\}-R_{*}
\end{multline*}
with $C_{\mu}=2^{\max\{\mu-2,0\}}\max\{\mu,\mu(\mu-1)\}$ as in Lemma~\ref{Lem:functional:inequality:1:b}. Thus, the claim follows since $\kappa=\max\{\kappa_1,\kappa_2\}$.
\end{proof}

Based on \cref{Lem:functional:inequality:1:b,Lem:functional:inequality:2}, we can now also show the uniqueness of solutions to~\eqref{eq:Smol:1}.

\begin{proof}[Proof of Proposition~\ref{Prop:uniqueness}]
 If the assumptions of Theorem~\ref{Thm:existence:stat:sol} are satisfied, the uniqueness of solutions to~\eqref{eq:Smol:1} follows immediately from Lemma~\ref{Lem:functional:inequality:2} and Grönwall's inequality. 
 
 On the other hand, if $\beta<\gamma$ we can argue similarly by means of Lemma~\ref{Lem:functional:inequality:1:b}. In fact, Lemma~\ref{Lem:functional:inequality:1:b} with $\mu=1$ yields
 \begin{equation*}
  \frac{\dd}{\dd{t}}\sum_{k=1}^{\infty}k\abs{c_{k}-d_{k}}\leq \biggl(6A_{*}\sum_{\ell=1}^{\infty}\ell^{1+\beta}(c_{\ell}+d_{\ell})-R_{*}\biggr)\sum_{k=1}^{\infty}k\abs{c_{k}-d_{k}}.
 \end{equation*}
 Using $R_{*}>0$ and Lemma~\ref{Lem:moment:general} we further get
  \begin{equation*}
  \frac{\dd}{\dd{t}}\sum_{k=1}^{\infty}k\abs{c_{k}-d_{k}}\leq 12C_{1+\beta}A_{*}(1+t^{-\beta/\gamma})\sum_{k=1}^{\infty}k\abs{c_{k}-d_{k}}.
 \end{equation*}
 Since $t\mapsto(1+t^{-\beta/\gamma})$ is integrable at zero, we obtain again uniqueness due to Grönwall's inequality.
\end{proof}

\section{Existence and convergence to equilibrium}\label{Sec:convergence}

Due to Lemma~\ref{Lem:functional:inequality:2}, we are now in a position to show the existence of a unique equilibrium and the convergence to it. The arguments in this section essentially only rely on~\eqref{eq:functional:inequality} and consequently, we can follow the same approach as in~\cite[Section~3]{FoM04}. However, for convenience and completeness we recall the proofs again in the following.

\begin{lemma}\label{Lem:zero:derivative}
Under the assumptions of Theorem~\ref{Thm:existence:stat:sol} we have
\begin{equation*}
 \lim_{t\to\infty}\abs*{\frac{\dd}{\dd{t}} c_{k}(t)}=0\quad \text{for all }k\in\N
\end{equation*}
and for each solution $c=(c_{k})_{k\in\N}$ of~\eqref{eq:Smol:1}.
\end{lemma}

\begin{proof}
  For each $h>0$ we define the shifted sequence $c^{h}=(c_{k}^{h})_{k\in \N}$ through
  \begin{equation*}
   c_{k}^{h}(t)=c_{k}(t+h)
  \end{equation*}
 and note that $c^{h}$ is again a solution to~\eqref{eq:Smol:1} with initial condition $c^{h,\text{in}}=c(h)$. Moreover, applying Lemma~\ref{Lem:functional:inequality:2} with the pair of solutions $c$ and $c^{h}$ we obtain
 \begin{equation*}
  \frac{\dd}{\dt}\sum_{k=1}^{\infty}k^{\mu}\abs{c_{k}^{h}-c_{k}}\leq \kappa \sum_{k=1}^{\infty}k^{\mu}\abs{c_{k}^{h}-c_{k}}\quad \text{for }t\geq T_{*}.
 \end{equation*}
 Integrating this differential inequality we find together with the definition of $c^{h}$ that
 \begin{equation*}
  \sum_{k=1}^{\infty}k^{\mu}\abs{c_{k}(t+h)-c_{k}(t)}\leq \sum_{k=1}^{\infty}k^{\mu}\abs{c_{k}(T_{*}+h)-c_{k}(T_{*})}\ee^{-\kappa(t-T_{*})}\quad \text{for }t\geq T_{*}.
 \end{equation*}
 Thus, since $c\in C^{1}((0,\infty),\ell_{\mu}^{1})$ we conclude for each $t\geq T_{*}$ that
 \begin{multline*}
  \norm*{\frac{\dd}{\dt}c(t)}_{\ell_{\mu}^{1}}=\lim_{h\to 0}\frac{1}{h}\sum_{k=1}^{\infty}k^{\mu}\abs{c_{k}(t+h)-c_{k}(t)}\\*
  \leq \lim_{h\to 0}\sum_{k=1}^{\infty}k^{\mu}\abs{c_{k}(T_{*}+h)-c_{k}(T_{*})}\ee^{-\kappa(t-T_{*})}=\norm*{\frac{\dd}{\dt}c(T_{*})}_{\ell_{\mu}^{1}}\ee^{-\kappa(t-T_{*})}.
 \end{multline*}
 Taking the limit $t\to\infty$ in this estimate we finally end up with
 \begin{equation*}
  \lim_{t\to\infty}\norm*{\frac{\dd}{\dt}c(t)}_{\ell_{\mu}^{1}}=0
 \end{equation*}
 from which the claim immediately follows.
\end{proof}

We are now in a position to give the proof of the existence of a unique equilibrium for~\eqref{eq:Smol:1} as well as the exponential  convergence to it.

\begin{proof}[Proof of Theorem~\ref{Thm:existence:stat:sol}]
 Let $c=(c_{k})_{k\in\N}$ be a solution to~\eqref{eq:Smol:1} with initial condition $c^{\text{in}}\equiv 0$. Due to \cref{Lem:total:mass:1,Lem:moment:general} we have for each $\mu\geq 1$ that
 \begin{equation}\label{eq:ex:equilibrium:1}
  \sup_{t\in[0,\infty)}\mom_{1}(t)\leq \widehat{\s}_{1}\quad \text{and}\quad \mom_{\mu}(t)\leq C_{\mu}(1+t^{-\rho_{\mu}})\quad \text{for all }t>0.
 \end{equation}
 Thus, by a standard diagonal argument there exists an increasing sequence $(t_{n})_{n\in\N}$ satisfying $t_{0}>0$ and $t_{n}\to\infty$ and moreover a sequence $Q=(Q_{k})_{k\in\N}$ such that 
 \begin{equation*}
  c(t_{n})\longrightarrow Q\quad \text{in } \ell_{\mu}^{1}\text{ for all }\mu\geq 1\quad \text{as } n\longrightarrow\infty.
 \end{equation*}
 In particular, \eqref{eq:ex:equilibrium:1} together with Fatou's Lemma also shows that 
 \begin{equation*}
  \sum_{k=1}^{\infty}kQ_{k}\leq \widehat{\s}_{1}\quad \text{and}\quad \sum_{k=1}^{\infty}k^{\mu}Q_{k}\leq C_{\mu} \quad \text{for all }\mu\geq 1.
 \end{equation*}
 Finally, applying Proposition~\ref{Prop:convergence:operator} to the stationary sequence $c^{n}\vcc=c(t_{n})$ we get that $Q$ is in fact a stationary solution to~\eqref{eq:Smol:1}. Uniqueness is then a direct consequence of~\eqref{eq:functional:inequality}.
\end{proof}

\begin{proof}[Proof of Theorem~\ref{Thm:convergence}]
 According to Lemma~\ref{Lem:functional:inequality:2} we have 
 \begin{equation*}
  \frac{\dd}{\dt}\sum_{k=1}^{\infty}k^{\mu}\abs{c_{k}-Q_{k}}\leq -\kappa\sum_{k=1}^{\infty}k^{\mu}\abs{c_{k}-Q_{k}}\quad \text{for }t\geq T_{*}.
 \end{equation*}
 Thus, we obtain by integration that
 \begin{equation}\label{eq:convergence:1}
  \sum_{k=1}^{\infty}k^{\mu}\abs{c_{k}(t)-Q_{k}}\leq \sum_{k=1}^{\infty}k^{\mu}\abs*{c_{k}(T_{*})-Q_{k}}\ee^{-\kappa (t-T_{*})}\quad \text{for } t\geq T_{*}.
 \end{equation}
 Moreover, we can fix $\tilde{\mu}\geq \mu$ such that $\tilde{\mu}>\max\{2-\alpha-\beta,1\}$. Then \cref{Rem:moments:monotonicity,Lem:moment:general:large:time} imply that there exists $C_{\tilde{\mu}}>0$ which only depend on the parameters in the equation and a time $\tilde{T}>0$ which also depends on $c$ such that $\sum_{k=1}^{\infty}k^{\mu}\abs{c_{k}(t)-Q_{k}}\leq \sum_{k=1}^{\infty}k^{\tilde{\mu}}(c_{k}(t)+Q_{k})\leq C_{\tilde{\mu}}$ if $t\geq \tilde{T}$. The claim thus follows if we choose $T_{c}=\max\{T_{*},\tilde{T}\}$.
\end{proof}

\section{An instructive example}\label{Sec:example}

In this section we provide an example which illustrates that on the one hand there is reason to believe that the result of Theorem~\ref{Thm:convergence} holds for a broader class of coefficients while on the other hand it seems impossible to show this by the method that we used above. We note that this example is only for illustration and there is probably no application behind it. We choose the coagulation kernel $a_{k,\ell}$ and the coefficients $r_{k}$ and $s_{k}$ such that
\begin{equation*}
 a_{k,\ell}=\begin{cases}
             A_{*} &\text{if }k=\ell=1\\
             0 &\text{otherwise},
            \end{cases}
            \qquad 
            r_{k}=R_{*}k^{\gamma}\quad \text{with } \gamma,R_{*}>0\qquad \text{and}\qquad s_{k}\geq 0 \quad \text{for all }k\in\N.
\end{equation*}
This leads to the following system of equations
\begin{equation}\label{eq:example:evolution}
 \begin{aligned}
  \frac{\dd}{\dt} c_{1}&=-A_{*} c_{1}^{2}+s_{1}-r_{1}c_{1}\\
  \frac{\dd}{\dt} c_{2}&=\frac{A_{*}}{2}c_{1}^{2}+s_{2}-r_{2}c_{2}\\
  \frac{\dd}{\dt} c_{k}&=s_{k}-r_{k}c_{k}  & \text{for }k\geq 3.
 \end{aligned}
\end{equation}
The corresponding system for the stationary state reads
\begin{equation}\label{eq:example:stationary}
 \begin{aligned}
  0&=-A_{*} Q_{1}^{2}+s_{1}-r_{1}Q_{1}\\
  0&=\frac{A_{*}}{2}Q_{1}^{2}+s_{2}-r_{2}Q_{2}\\
  0&=s_{k}-r_{k}Q_{k}  && \text{for }k\geq 3.
 \end{aligned}
\end{equation}
From this, we immediately obtain
\begin{equation*}
 Q_{k}=\frac{s_{k}}{r_{k}}\qquad \text{for }k\geq 3.
\end{equation*}
Moreover, the solutions to $0=-A_{*} Q_{1}^{2}+s_{1}-r_{1}Q_{1}$ are given by
\begin{equation*}
 Q_{1}^{+}=\frac{1}{2}\sqrt{\frac{r_{1}^{2}}{A_{*}^2}+4\frac{s_{1}}{A_{*}}}-\frac{r_{1}}{2A_{*}}\quad \text{and}\quad Q_{1}^{-}=-\frac{1}{2}\sqrt{\frac{r_{1}^{2}}{A_{*}^2}+4\frac{s_{1}}{A_{*}}}-\frac{r_{1}}{2A_{*}}
\end{equation*}
while only $Q_{1}^{+}$ has a positive sign. Thus, the only possible choice for the first component $Q_{1}$ of the equilibrium $Q=(Q_{k})_{k\in\N}$ is $Q_{1}=Q_{1}^{+}$. Finally, the unique solution of $0=\frac{A_{*}}{2}Q_{1}^{2}+s_{2}-r_{2}Q_{2}$ is given by
\begin{equation*}
 Q_{2}=\frac{\frac{A_{*}}{2}Q_{1}^{2}+s_{2}}{r_2}=\frac{A_{*}}{2r_{2}}Q_{1}^{2}+\frac{s_{2}}{r_{2}}.
\end{equation*}
Thus, in summary the unique equilibrium of~\eqref{eq:example:evolution}, i.e.\@ the unique solution to~\eqref{eq:example:stationary} is
\begin{equation}\label{eq:explicit:equilibrium}
 Q_{1}=\frac{1}{2}\sqrt{\frac{r_{1}^{2}}{A_{*}^2}+4\frac{s_{1}}{A_{*}}}-\frac{r_{1}}{2A_{*}},\quad Q_{2}=\frac{A_{*}}{2r_{2}}Q_{1}^{2}+\frac{s_{2}}{r_{2}}\quad \text{and}\quad Q_{k}=\frac{s_{k}}{r_{k}}\qquad \text{for }k\geq 3.
\end{equation}
We compute now the solution to~\eqref{eq:example:evolution} explicitly. For this let~\eqref{eq:example:evolution} be equipped with the initial condition $c^{\text{in}}=(c_{k}^{\text{in}})_{k\in\N}$. We immediately find that
\begin{equation}\label{eq:ck:explicit}
 c_{k}(t)=c_{k}^{\text{in}}\ee^{-r_{k} t}+\frac{s_{k}}{r_{k}}(1-\ee^{-r_{k}t}) \quad \text{for }k\geq 3
\end{equation}
is the unique solution to $\frac{\dd}{\dt} c_{k}=s_{k}-r_{k}c_{k}$. On the other hand, $\frac{\dd}{\dt} c_{1}=-A_{*} c_{1}^{2}+s_{1}-r_{1}c_{1}$ is a Riccati equation with constant coefficients which is well-known to have the unique solution
\begin{equation}\label{eq:explicit:c1}
 c_{1}(t)=Q_{1}+\frac{\frac{1}{\alpha}\frac{c_{1}^{\text{in}}-Q_{1}}{c_{1}^{\text{in}}-Q_{1}^{-}}}{1-\left(1-\frac{1}{\alpha(c_{1}^{\text{in}}-Q_{1}^{-})}\right)\ee^{-\frac{A_{*}}{\alpha}t}}\ee^{-\frac{A_{*}}{\alpha}t}\qquad\text{with }\alpha\vcc=\frac{1}{Q_{1}-Q_{1}^{-}}=\Bigl(\frac{r_{1}^{2}}{A_{*}^2}+4\frac{s_{1}}{A_{*}}\Bigr)^{-\frac{1}{2}}.
\end{equation}
Since $\frac{\dd}{\dt} c_{2}=\frac{A_{*}}{2}c_{1}^{2}+s_{2}-r_{2}c_{2}$ is linear with respect to $c_{2}$, we can immediately write down a solution formula for $c_{2}$ in terms of $c_{1}$. Precisely, we have
\begin{equation}\label{eq:c2:explicit}
 c_{2}(t)=c_{2}^{\text{in}}\ee^{-r_{2}t}+\frac{s_{2}}{r_{2}}(1-\ee^{-r_{2}t})+\int_{0}^{t}\frac{A_{*}}{2}c_{1}^{2}(s)\ee^{-r_{2}(t-s)}\ds.
\end{equation}
 Moreover, it is straightforward to estimate
\begin{equation*}
 1-\left(1-\frac{1}{\alpha(c_{1}^{\text{in}}-Q_{1}^{-})}\right)\ee^{-\frac{A_{*}}{\alpha}t}\geq \min\Bigl\{1,\frac{1}{\alpha(c_{1}^{\text{in}}-Q_{1}^{-})}\Bigr\}\quad \text{for all }t\geq 0.
\end{equation*}
 Thus, we obtain immediately from~\eqref{eq:explicit:c1} that
 \begin{equation}\label{eq:explicit:c1:2}
  \abs{c_{1}(t)-Q_{1}}\leq \frac{\frac{1}{\alpha}\abs*{\frac{c_{1}^{\text{in}}-Q_{1}}{c_{1}^{\text{in}}-Q_{1}^{-}}}}{\min\bigl\{1,\frac{1}{\alpha(c_{1}^{\text{in}}-Q_{1}^{-})}\bigr\}}\ee^{-\frac{A_{*}}{\alpha}t}\leq \max\Bigl\{\frac{1}{\alpha}\abs[\Big]{\frac{c_{1}^{\text{in}}-Q_{1}}{c_{1}^{\text{in}}-Q_{1}^{-}}},\abs{c_{1}^{\text{in}}-Q_{1}}\Bigr\}\ee^{-\frac{A_{*}}{\alpha}t}.
 \end{equation}
 Note that we additionally exploited that $c_{1}^{\text{in}}\geq 0$ and $Q_{1}^{-}<0$. In fact, the latter also implies $\abs{c_{1}^{\text{in}}-Q_{1}^{-}}\geq \abs{Q_{1}^{-}}$ such that the right-hand side of~\eqref{eq:explicit:c1:2} can be further estimated as
\begin{equation}\label{eq:convergence:c1}
 \abs{c_{1}(t)-Q_{1}}\leq \max\bigl\{1,(\alpha\abs{Q_{1}^{-}})^{-1}\bigr\}\abs{c_{1}^{\text{in}}-Q_{1}}\ee^{-\frac{A_{*}}{\alpha}t}.
\end{equation}
The previous estimate yields in particular $(c_{1}(t)+Q_{1})\leq \max\bigl\{1,(\alpha\abs{Q_{1}^{-}})^{-1}\bigr\}\abs{c_{1}^{\text{in}}-Q_{1}}+2Q_{1}$ for all $t\geq 0$. Using this together with \cref{eq:c2:explicit,eq:explicit:equilibrium,eq:convergence:c1} we further deduce that
\begin{align*}
 \abs*{c_{2}(t)-Q_{2}}&\leq \Bigl(c_{2}^{\text{in}}+\frac{s_{2}}{r_{2}}\Bigr)\ee^{-r_{2}t}+\frac{A_{*}}{2}\int_{0}^{t}(c_{1}(s)+Q_{1})\abs{c_{1}(s)-Q_{1}}\ee^{-r_{2}(t-s)}\ds+\frac{A_{*}}{2r_{2}}Q_{1}^{2}\ee^{-r_{2}t}\\
 &\leq \Bigl(c_{2}^{\text{in}}+\frac{s_{2}}{r_{2}}+\frac{A_{*}}{2r_{2}}Q_{1}^{2}\Bigr)\ee^{-r_{2}t}\\
 &\qquad+\frac{A_{*}}{2}\bigl(\max\bigl\{1,(\alpha\abs{Q_{1}^{-}})^{-1}\bigr\}\abs{c_{1}^{\text{in}}-Q_{1}}+2Q_{1}\bigr)\abs{c_{1}^{\text{in}}-Q_{1}}\int_{0}^{t}\ee^{\left(r_{2}-\frac{A_{*}}{\alpha}\right)s}\ds\ee^{-r_{2}t}.
\end{align*}
One can check by straightforward estimates that $\int_{0}^{t}\ee^{\left(r_{2}-\frac{A_{*}}{\alpha}\right)s}\ds\ee^{-r_{2}t}\leq t\ee^{-\min\{A_{*}/\alpha,r_{2}\}t}$. With this, we easily deduce $\int_{0}^{t}\ee^{\left(r_{2}-\frac{A_{*}}{\alpha}\right)s}\ds\ee^{-r_{2}t}\leq \frac{2}{\ee\min\{A_{*}/\alpha,r_{2}\}}\ee^{-\min\{A_{*}/(2\alpha),r_{2}/2\}t}$ from which we thus conclude
\begin{multline}\label{eq:convergence:c2}
 \abs*{c_{2}(t)-Q_{2}}\leq\biggl(c_{2}^{\text{in}}+\frac{s_{2}}{r_{2}}+\frac{A_{*}}{2r_{2}}Q_{1}^{2}\\*
 +\frac{A_{*}}{\ee\min\{A_{*}/\alpha,r_{2}\}}\bigl(\max\bigl\{1,(\alpha\abs{Q_{1}^{-}})^{-1}\bigr\}\abs{c_{1}^{\text{in}}-Q_{1}}+2Q_{1}\bigr)\abs{c_{1}^{\text{in}}-Q_{1}}\biggr)\ee^{-\min\left\{\frac{A_{*}}{2\alpha},\frac{r_{2}}{2}\right\}t}.
\end{multline}
From~\cref{eq:ck:explicit,eq:explicit:equilibrium} we directly get
\begin{equation}\label{eq:convergence:ck}
 \abs*{c_{k}(t)-Q_{k}}\leq \Bigl(c_{k}^{\text{in}}+\frac{s_{k}}{r_{k}}\Bigr)\ee^{-r_{k}t}\quad\text{for all }k\geq 3.
\end{equation}
Due to the assumptions on $s_{k}$ and $r_{k}$ in \cref{eq:Ass:r,eq:Ass:s} we obtain for each $\mu\geq 1$ that 
\begin{multline}\label{eq:convergence:ck:2}
 \sum_{k=3}^{\infty}k^{\mu}\abs{c_{k}-Q_{k}}\leq \sum_{k=3}^{\infty}\Bigl(c_{k}^{\text{in}}+\frac{s_{k}}{r_{k}}\Bigr)\ee^{-r_{k}t}\leq \sum_{k=3}^{\infty}k\Bigl(c_{k}^{\text{in}}+\frac{s_{k}}{R_{*}}\Bigr)\bigl(k^{\mu-1}\ee^{-\frac{R_{*}}{2}k^{\gamma}t}\bigr)\ee^{-\frac{R_{*}}{2}k^{\gamma}t}\\*
 \leq C_{\mu,\gamma}\Bigl(\mom_{1}^{\text{in}}+\frac{\s_{1}}{R_{*}}\Bigr)R_{*}^{\frac{1-\mu}{\gamma}}t^{-\frac{\mu-1}{\gamma}}\ee^{-\frac{3^{\gamma}}{2}R_{*}t}.
\end{multline}
In the last step we used that $(k^{\mu-1}\ee^{-\frac{R_{*}}{2}k^{\gamma}t})\leq C_{\mu,\gamma} (R_{*}t)^{(1-\mu)/\gamma}$. Summarising \cref{eq:convergence:c1,eq:convergence:c2,eq:convergence:ck:2} we thus obtain that
\begin{equation}\label{eq:example:norm:convergence}
 \sum_{k=1}^{\infty}k^{\mu}\abs{c_{k}-Q_{k}}\leq K\ee^{-\kappa t} \quad \text{for all }t\geq 1\quad \text{with }\kappa=\min\Bigl\{\frac{A_{*}}{2\alpha},\frac{r_{2}}{2},\frac{3^{\gamma}}{2}R_{*}\Bigr\}.
\end{equation}
We recall from \cref{Lem:functional:inequality:1:b,Lem:functional:inequality:2} and the corresponding proofs that, in order to apply the abstract approach from \cref{Sec:convergence,Sec:func:inequality} which shows convergence to the equilibrium, we require that 
\begin{equation*}
 2A_{*}(C_{\mu}+2)\sum_{\ell=1}^{\infty}\ell^{\mu+\beta}(c_{\ell}+Q_{\ell})-R_{*}\qquad \text{with } C_{\mu}=2^{\max\{\mu-2,0\}}\max\{\mu,\mu(\mu-1)\}
\end{equation*}
can be estimate from above by a negative constant. In contrast to this, we will now show that for a specific choice of the parameters this quantity is strictly positive while nevertheless~\eqref{eq:example:norm:convergence} still yields exponential convergence to the equilibrium. This example thus illustrates that the global method that we used seems to be limited to a certain range of parameters. In fact using the non-negativity of $c$ and $Q$, as well as the explicit form of $Q_{1}$ and choosing $r_{k}=R_{*}k^{\gamma}$ for all $k\in\N$ we find
\begin{multline*}
 2A_{*}(C_{\mu}+2)\sum_{\ell=1}^{\infty}\ell^{\mu+\beta}(c_{\ell}+Q_{\ell})-R_{*}\geq 2(C_{\mu}+2)A_{*} Q_{1}-R_{*}\\*
 =(C_{\mu}+2)\left(\sqrt{r_{1}^{2}+4A_{*}s_{1}}-r_{1}\right)-R_{*}=R_{*}\left[(C_{\mu}+2)\left(\sqrt{1+\frac{4A_{*}s_{1}}{R_{*}^{2}}}-1\right)-1\right].
\end{multline*}
If we choose now the constants $A_{*}$ and $s_{1}$ such that $A_{*}s_{1}\geq 4R_{*}^{2}$ we further obtain
\begin{equation*}
 2A(C_{\mu}+2)\sum_{\ell=1}^{\infty}\ell^{\mu+\beta}(c_{\ell}+Q_{\ell})-R_{*}\geq R_{*}\bigl[(C_{\mu}+2)(4-1)-1\bigr]=\bigl(3(C_{\mu}+2)-1)R_{*}>0.
\end{equation*}

\textbf{Acknowledgments:} CK and ST have been supported by a Lichtenberg 
Professorship of the VolkswagenStiftung.

\end{document}